\documentclass[11pt]{article}

\usepackage{amsmath, amsfonts, latexsym, amsthm, amssymb, verbatim, tikz, cite, array,appendix,pgfplots}

\usetikzlibrary{calc} 
\usetikzlibrary{shapes}
\usetikzlibrary{arrows} 
\usetikzlibrary{through} 
\usepackage[mathcal]{euscript}
\usepackage{color,graphicx}

\definecolor{darkgreen}{rgb}{0,0.6,0}

\newcommand{\Z}{{\mathbb Z}}
\newcommand{\ZN}{\Z_N}
\newcommand{\Hist}[1]{{\mathcal H_{#1}}}
\newcommand{\Histg}[1]{H_{#1}}

\newcommand{\IS}{{\mathcal I}}
\newcommand{\ISp}[1]{\IS^{{#1}+}}
\newcommand{\ISm}[1]{\IS^{{#1}-}}
\newcommand{\ip}{{\mathcal S}}
\newcommand{\cp}{{\mathcal C}}
\newcommand{\obs}{{\mathcal Z}}
\newcommand{\+}{{\text{ }\cap\text{ }}}

\numberwithin{equation}{section}
\numberwithin{figure}{section}
\numberwithin{table}{section}
\newtheorem{thm}{Theorem}[section]
\newtheorem{cor}[thm]{Corollary}
\newtheorem{lem}[thm]{Lemma}

\newtheorem{hyp}[thm]{Hypothesis}
\newcommand{\Prob}{\mathbb{P}} 
\newcommand{\E}{\mathbb{E}} 
\newcommand{\pmfp}[1]{F_{#1}^+}
\newcommand{\pmfm}[1]{F_{#1}^-}

\newcommand{\pmf}{{\mathcal P}(\ZN)}

\newcommand{\argmax}[1]{\underset{#1}{\operatorname{arg\,max\,}}}
\newcommand{\argmin}[1]{\underset{#1}{\operatorname{arg\,min\,}}}

\usepackage{multicol}

\newcommand{\allhist}{\Hist{}}
\newcommand{\ZZ}{\mathbf{Z}}

\title{Information Maximization Fails to Maximize Expected Utility in a Simple Foraging Model}
\author{Edward K. Agarwala$^{1,2}$, Hillel J. Chiel$^{3,4,5}$, and Peter J. Thomas$^{1,3,6,7}$.\\
$^{1-6}$Case Western Reserve University, Departments of \\ $^1$Mathematics, $^2$Operations Reserach, $^3$Biology,\\ $^4$Neurosciences, $^5$Biomedical Engineering, and $^6$Cognitive Science.\\
$^7$Oberlin College Department of Neuroscience.}

\begin{document}
\maketitle
\begin{abstract}
Information theory has successfully explained the organization of many biological phenomena, from the physiology of sensory receptive fields to the variability of certain DNA sequence ensembles.  Some scholars have proposed that information should provide the \textit{central} explanatory principle in biology, in the sense that any behavioral strategy that is optimal for an organism's survival must necessarily involve  efficient information processing.  Here we challenge this view by providing a counterexample.  We present an analytically tractable model for a particular instance of a perception-action loop: a creature searching for a wandering food source confined to a one-dimensional ring world.  The model incorporates the statistical structure of the creature's world, the effects of the creature's actions on that structure, and the creature's strategic decision process.  The underlying model takes the form of a Markov process on an infinite dimensional state space.  To analyze it  we construct an exact coarse graining that reduces the model  to a Markov process on a finite number of ``information states".  This mathematical technique allows us to make quantitative comparisons between the performance of an information-theoretically optimal strategy with other candidate search strategies on a food gathering task.  We find that 
\begin{enumerate}
\item Information optimal search does not necessarily optimize utility (expected food gain).
\item The rank ordering of search strategies by information performance does not predict their ordering by expected food obtained.
\item The relative advantage of different strategies depends on the statistical structure of the environment, in particular the variability of motion of the source.
\end{enumerate} 
We conclude that \emph{there is no simple relationship between information and utility}. Behavioral optimality does not imply information efficiency, nor is there a simple tradeoff between the two objectives of gaining information about a food source \textit{versus} obtaining the food itself.

\end{abstract}

\section{Introduction}

\subsection{Information optimality in sensory systems}
It has long been recognized that information efficiency is an important explanatory principle underlying the function of biological sensory systems \cite{Attneave1954,Barlow1961}.  Linsker \cite{Linsker1988Computer} proposed that an information maximization  principle might govern the structure of  layered self-adaptive neural networks such as those found in the visual system.  It was subsequently shown that efficient encoding of the visual and auditory environment  could account for receptive field physiology observed  in cats and macaques \cite{van1998independent,Smith-Lewicki-06}. Analysis of information optimality in sensory systems typically proceeds  independently of considerations of the relative biological significance of different stimuli in the environment.  Simoncelli and Olshausen \cite{Simoncelli+Olshausen:2001:AnnRevNsci} noted that the principle of efficient coding 
succeeds in part because it does \emph{not} take into account complicating factors such as the accuracy with which signals are represented, or the costs to the organism of making mistakes.

Several authors have extended the framework of information theory to take into account the biological context of  behaving organisms.
Shraiman and colleagues, for instance, coined the term \emph{infotaxis} \cite{vergassola2007infotaxis} to describe  search strategies that locally maximize the expected rate of information gain by an organism.
In the context of a model organism pursuing the source of a diffuse pheromone plume, these authors showed that infotaxis outperformed chemotaxis (locally maximizing the expected rate of increase in signal concentration) with respect to mean latency of capture.  

Information theoretic arguments have succeeded in molecular biology as well \cite{OfriaAdamiCollier:2003:JTB}. Schneider, for instance, demonstrated that the statistical entropy of nucleic acid sequences encoding location information read out by proteins docking at specific DNA sites can be  predicted by the amount of information required to specify the docking sites among all possible binding sites, \textit{i.e.}~the mean entropy decrease from the undocked to the docked state \cite{Schneider1986JMolBiol,Schneider:2000:NAR}.  Moreover the uncertainty decrease in such ``molecular machines'' can be related \emph{quantitatively} to the free energy of protein--DNA binding \cite{Schneider:2010:NAR}.

These results underscore the importance of information theory for understanding sensory, behavioral and genetic systems.  Some authors go so far as to suggest that information theory could serve as an overriding explanatory principle throughout biology.  For example, Adami writes
\begin{quotation}
The discovery of the genetic code cemented the fact that information is the \emph{central} 
pillar in any attempt to understand life, and the dynamics of information storage and acquisition that come with it. [\cite{adami1998ial}, pg. 59, original emphasis].
\end{quotation}
In a similar spirit Polani has proposed parsimonious information processing 
as a universality principle in biology,
arguing that ``if organisms would develop a suboptimal information-processing 
strategy, this would waste metabolic energy. Such a disadvantage would then be selected against by evolution.'' [\cite{Polani2009HFSP}, p.3].
For the sake of argument we will state what we call the \emph{strong hypothesis}: 
\begin{hyp} \label{hyp:strong}
Behaviors that are optimal for an organism's survival are necessarily information optimal.
\end{hyp}
Evaluating such a hypothesis requires precise definitions of behavioral and information optimality.  On its face,
the hypothesis is plausible, inasmuch as it would appear that any successful survival strategy that led to suboptimal information processing should be capable of improvement through improved efficiency of information processing. However, it is possible that  complicating factors such as metabolic constraints or the costs of acquiring and processing information may play a role on a par with that of information maximization \textit{per se}.   
We undertook a quantitative analysis of the relationship between information optimality and utility by constructing an analytically tractable model of a creature using sensory information to search for a source of food, taking care to allow for precise definitions of
behavioral and information optimality.  As we shall show in the sequel, we find that \emph{even in the absence of metabolic and computational constraints,}  information-optimal behavior and utility-optimal behavior bear no simple relation to one another.  

\subsection{Perception-Action Loop}

Studies of the interplay of sensory processing and behavior are complicated by the structure of the problem, which necessarily embodies a closed ``perception-action loop".  
Barlow's redundacy reduction hypothesis \cite{Barlow1961} led to models of sensory processing within the framework now known as acyclic graphical models \cite{Jordan1998}.  
In such  models the pattern of statistical dependencies form an unambiguous chain or tree from causes (\textit{e.g.}~environmental signals) to effects (\textit{e.g.}~activity in a layer of sensory neurons). 
The analysis of information efficiency in this context translates into statements about chains of self consistent conditional probabilities relating the distributions of random variables defined on the nodes of the acyclic graph; one obtains the optimal architecture  to accomplish a processing goal (\textit{e.g.}~output decorrelation) given the statistical structure of the sensory world.   But although modeling sensory perception may be approached along such lines, allowing for a creature's behavior requires a fundamentally different probabilistic framework.  

\begin{figure}[htbp] 
   \centering
   \includegraphics[width=4in]{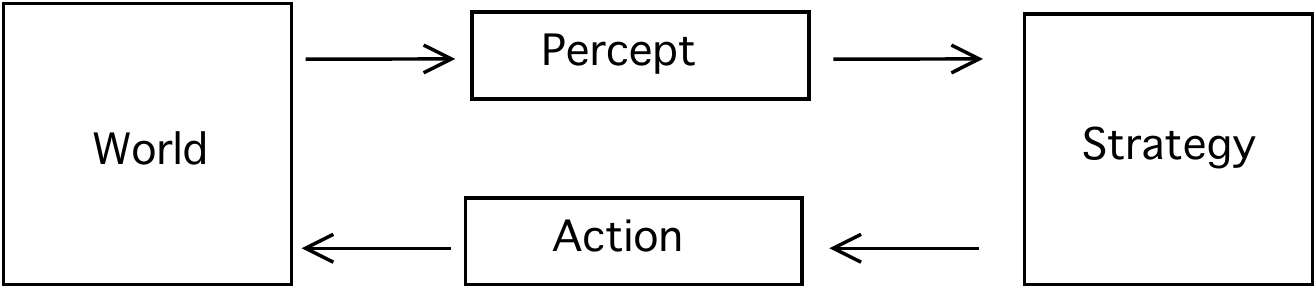} 
   \caption[The perception-action loop.]{The perception-action loop.  The statistical structure of the world determines the percept; the creature determines an action by virtue of a strategy (taking into account the percept); the action changes the statistical structure of the world, relative to the creature.}
   \label{fig:loop}
\end{figure}

Consider a scenario in which a creature observes the environment in order to accomplish some survival related goal, such as locating and exploiting a food source.  Having observed the available stimuli a creature may take action that in turn alters the statistical structure of the world, as seen from  the creature's perspective.  In the simplest case, the movement of the creature towards or away from the source changes the statistics of the target related stimuli \footnote{An analogous situation arises in saccadic visual search, because the spatial resolution of the retina varies significantly from the foveal to the peripheral region.  Making a saccade both gathers additional information about possible target locations and changes the statistical distribution of the inputs relative to the (unknown) target location \cite{Najemnik+Geisler:2005:Nature}.
 In both that situation and in the model considered here, the problem becomes tractable when  placed in an ideal observer framework \cite{GeislerChapter52:2004}. }.  Under these circumstances the conditional probabilities describing the stimulus (the world), the received signal (a percept), the creature's decision about how to proceed (a strategy) and the effects of the creature's behavior (an action) form a closed loop (Figure \ref{fig:loop}).   Perception-action loops appear in models of behavior ranging from  simplified computational models of evolving agents \cite{KlyubinPolaniNehaniv:2007:NeuralComp} to conceptual models of cognitive processing in the mammalian prefrontal cortex\footnote{Compare Figure 1 of \cite{Fuster:2004:TrendsCogSci} with Figure \ref{fig:loop} of the present manuscript.} \cite{Fuster:2004:TrendsCogSci},
 as well as systems in which the brain, the body and the environment coevolve \cite{Chiel+Beer:1997:TrendsNsci}.
 The importance of feedback loops reciprocally connecting a creature's ``input'' and ``output'' has been recognized since the earliest attempts to apply information theory to biology (\textit{cf} \cite{Ashby1956}, Chapter 3).   
 
 As is well known, statistical problems such as inference,  estimation and sampling, which  are well understood for probabilistic models on acyclic graphs, become essentially intractable for graphical models containing cyclic dependencies \cite{Wainwright2002,WainwrightJordan2006IEEE}.  The difficulty of analyzing probabilistic systems on graphs with cycles has been identified as a significant barrier to the progress of theory in both neuroscience and computational intelligence \cite{Bell1999PTRSLB}.
We overcome this obstacle for a relatively tractable model system by incorporating  
the entire system of interest into a larger Markov process for which the full transition matrix and steady state distribution can be obtained analytically.  To this end, we consider a simple model of an organism searching for a target, introduced below (Section \ref{ssec:2models}).  
We begin with a creature that can approximately estimate the location of a food source by measuring the (noisy) local food concentration at a sequence of locations.  The source itself moves unpredictably, following a random walk.  The combination of the creature and the world amounts to a Markov random walk in a large state space.  We then obtain a second, even more tractable model from the first through a heuristic limiting process, by reducing the noisiness of the food concentration measurements.  In the low noise limit we are able to reduce the entire system to a Markov process on a small number of states.  For the reduced system we then calculate the quantities of interest exactly: the mean amount of food gathered and the average uncertainty about the source location, and how these quantities depend on the creature's choice of search strategy.
We also vary a key parameter, the variability of the motion of the food source itself, and study its effect on the creature's performance under different strategies.




  An important feature of our system is the inclusion of an internal state in the model organism.  Whether they involve memory or metabolism, internal states play an important role in many models of biological decision making.  Life history theory, to cite one example, incorporated internal states as a form of phenotypic plasticity, thereby increasing its explanatory power.  Inclusion of information about the metabolic and environmental resources available to an organism improved understanding of decisions such as reproductive timing \cite{McNamara+Houston:1996:Nature}.  A similar distinction underlies the difference between chain reflex models of locomotion \cite{JindrichFull:2002JEB,ShikOrlovsky:1976:1976}, and robust interplay of a central pattern generator (internal state) guided by external information \cite{KukillayaProctorHolmes:2009:Chaos,MarkinKlishkoShevtsovaPrilutskyRybak2010AnnNYAS}.  
An analogous dichotomy appears in engineering and control theory, namely the distinction between open-loop control (throwing a baseball as precisely as possible at a target) and closed-loop control (flying an airplane with sensory feedback providing ongoing course corrections).  Notably, Touchette and Lloyd \cite{TouchetteLloyd:2000:PRL} analyzed both open- and closed-loop control in a general  information theoretic framework.  
In a control theoretic setting the requirement of controlling the state of a system while also obtaining information about unknown parameters describing it leads to the challenging ``dual-control" problem \cite{Griffiths+Loparo:1985:IntJCtrl,Feng+Loparo:1997:AutoControl-ieee,Feng+Loparo+Fang:1997:AutoControl-ieee}.





In the model we present here all four parts of a perception-action loop are represented in nontrivial terms, yet the entire system remains tractable enough to yield  a complete analytic treatment. 
This approach allows us to obtain quantitative results on the success of different search strategies.  We consider three movement algorithms, defined below (Section \ref{sec:strategies_intro}).  Each corresponds to a particular search strategy: one designed to minimize the creature's \emph{uncertainty about the location} of the food source (the ``Information Theory Creature", ITC), one designed to maximize the \emph{probability of colocation} of the creature with the food source (the ``Maximum Likelihood Creature", MLC) and a third strategy combining elements of the first two (the ``Hybrid" or ``modified Maximum Likelihood Creature", mMLC).  We studied the performance of each algorithm as a function of a parameter controlling the unpredictability of the source movement.  As expected, the ITC had more information about the source location on average than the other two, for all values of the source mobility parameter (Figure 2.9). 
Surprisingly, however, no one search strategy consistently dominated the others in terms of utility.  Each strategy earned a larger expected food yield than the other two for some range of mobility parameters (Figure \ref{fig:strategy_food}).    The ITC performed the best when the movement of the food source was the most random, yet the mMLC that used a mixture of this information-optimal strategy and the maximum likelihood based strategy performed the best when the source was the most predictable.  The MLC dominated the others for intermediate source mobility values. 

In short we show that for this tractable model system  there is \emph{no obvious simple relationship} between the information theoretic quality of a search strategy and its performance in terms of expected utility.  It therefore serves as a  counterexample against the strong hypothesis \ref{hyp:strong}. Our approach -- embedding the creature and its world into a single Markov process -- can be generalized in a variety of ways to incorporate additional biological details.  We consider several possible extensions in Section \ref{sec:conclusion}.



\subsection{Two Models of Search with Perfect Information}
\label{ssec:2models}

Appendix \ref{sec:notation} provides a summary of notation and abbreviations used throughout.

\subsubsection{Features common to both models}
In both models, the source (or injector) and the creature inhabit a discrete ring world with $N$ sites, identified with $\ZN$.  The source location $S_t\in\{0,\cdots,N-1\}$ is a random variable that is a function of discrete time $t\in\{0,1,2,\cdots\}$. Distance between locations on the ring is the minimum number of steps clockwise (CW) or counter-clockwise (CCW)  between the locations. Formally, the distance between locations $l$ and $m$ is equal to $\min\{{l-m \mod N}, {m-l\mod N\}}$ (see Figure \ref{fig:world}).\footnote{Addition and subtraction of positions on the ring are interpreted mod $N$ throughout.  By convention, we will enumerate loci on the circle clockwise. For simplicity we take $N$ to be even.}

\begin{figure}
\begin{center}
\begin{tikzpicture}[line cap=round,line width=1pt]
\draw (0,0) circle (1.3cm);
\foreach \angle in
{0, 20, 40, 60, 80, 100, 120, 140, 160, 180, 200, 220, 240, 260, 280, 300, 320, 340 }
{
\draw[line width=1pt] (\angle:1.4cm) -- (\angle:1.3cm);
}

\foreach \angle / \labl in
{40/$C$, 160/$S$}
{
\draw[line width=1pt] (\angle:1.7cm) -- (\angle:1.3cm);
\draw (\angle:1.9cm) node {\labl};
}

\draw (0,0) node {$\Z_N$};

\draw[color=blue,<->] (40:1.5cm) arc (40:160:1.5cm);
\draw[color=blue] (100:1.75) node {$D$};

\end{tikzpicture}
\end{center}

\caption[The World]{Geometry of the ring world, $\ZN$, occupied by the creature and the food source.  $C$ is the creature's location. $S$ is the source's location. $D$ is the distance between $C$ and $S$, taken to be $D = \min\{S-C \mod N, C-S\mod N\}$.}
\label{fig:world} 
\end{figure}
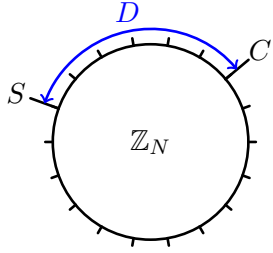

The source performs a random walk with uniformly distributed initial position $S_0$ and independent increments.
(Without loss of generality, we set the creature's initial position, $C_0$, to zero.)
The source's walk is governed by a probability mass function (PMF) $P$, defined such that the probability of the source  moving $m$ steps CW is $P(m)$.  
We denote the set of PMFs on $\ZN$, also known as the probability simplex on $\ZN$, as $\pmf$.
The probability mass function for the source location at time $t+1$, $S_{t+1}$ is obtained by convolving $P$ with the PMF for the current location  $S_t$. 
For simplicity we consider movements to the right or left of the current position equally likely, and thus source movement distributions take the form
\begin{equation}
P_x(m)=\left\{\begin{array}{ll}x,& m=1\\
1-2x,& m=0\\
x,&m=-1.  \end{array} \right.
\end{equation}
For technical reasons we will restrict the source mobility parameter $x$ to the range $1/4< x < 1/3$ (see Figure \ref{fig:information_states-}).


The world evolves following discrete time dynamics.  At each time $t$,  the source establishes a distribution of food molecules around the ring with a diffusion and decay process producing independent Poisson distributed molecule counts at each site.  The means are determined by an equilibrium condition balancing an injection rate $\gamma$, a decay rate $\alpha$ and a transition rate $y$ between adjacent nodes. 
%
%
Let $\lambda_m$,  $m\in\{0,\cdots,N-1\}$ be the mean number of molecules $m$ steps CW of the source, and define the vector of means $\Lambda=(\lambda_0,\cdots,\lambda_{N-1})$.
As calculated in Appendix \ref{app:food}, the vector $\Lambda$ is obtained (numerically) by solving a linear system of equations.
Figure \ref{fig:food-distribution} illustrates a typical food distribution; in this case
 $\alpha=.1$ and $y=1/3$.   
%


%
%


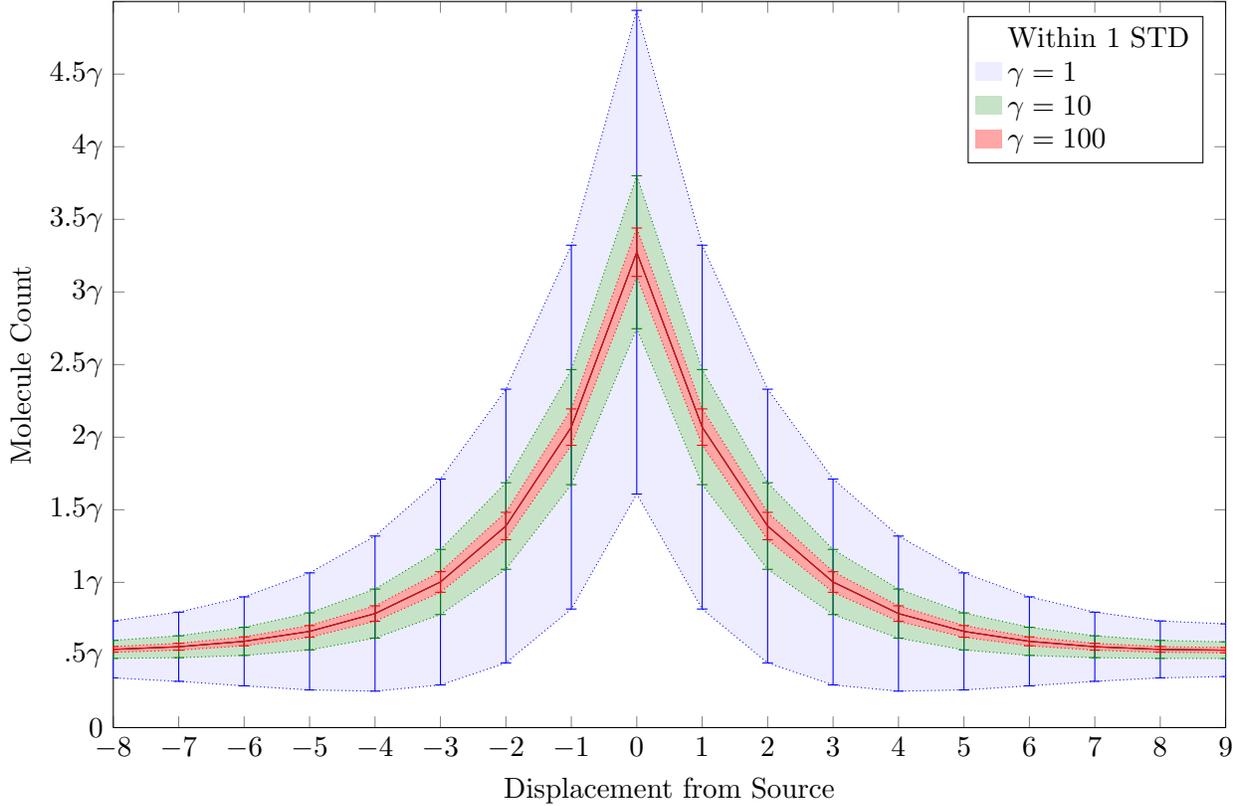
\begin{figure}

\begin{tikzpicture}

\definecolor{mycolor1}{rgb}{0,0.5,0}

\begin{axis}[%
scale only axis,
width=5.82222in,
height=3.80333in,
xmin=-8, xmax=9,
ymin=-0.5, ymax=4.5,
xlabel={Displacement from Source},
yticklabels={$-.5\gamma$,0, $.5\gamma$,$1\gamma$,$1.5\gamma$,$2\gamma$,$2.5\gamma$,$3\gamma$,$3.5\gamma$,$4\gamma$,$4.5\gamma$},
ylabel={Molecule Count},
every axis legend/.append style={nodes={right}},
legend entries={Within 1 STD,{$\gamma=1$},{$\gamma=10$},{$\gamma=100$}},
legend image code/.code={\fill[#1] (-.01,-.01) rectangle (0.2cm,0.15cm);}
]
\addplot[white] file {ErrorBar1}; 

\addplot[blue,fill=blue!7!white,densely dotted] file {ErrorBar1}; 
\addplot[mycolor1,fill=mycolor1!22!white,densely dotted] file {ErrorBar10};
\addplot[red,fill=red!35!white,densely dotted] file {ErrorBar100};

\addplot [
color=blue, solid
]
plot [error bars/.cd, y dir = both, y explicit]
coordinates{ (-8,0.0384222) +- (0,0.196016) (-7,0.0567185) +- (0,0.238156) (-6,0.093921) +- (0,0.306465) (-5,0.16243) +- (0,0.403027) (-4,0.285083) +- (0,0.533932) (-3,0.502764) +- (0,0.709059) (-2,0.888033) +- (0,0.942355) (-1,1.56931) +- (0,1.25272) (0,2.7737) +- (0,1.66544) (1,1.56931) +- (0,1.25272) (2,0.888033) +- (0,0.942355) (3,0.502764) +- (0,0.709059) (4,0.285083) +- (0,0.533932) (5,0.16243) +- (0,0.403027) (6,0.093921) +- (0,0.306465) (7,0.0567185) +- (0,0.238156) (8,0.0384222) +- (0,0.196016) (9,0.0329333) +- (0,0.181475)
};

\addplot [
color=mycolor1,
solid
]
plot [error bars/.cd, y dir = both, y explicit]
coordinates{ (-8,0.0384222) +- (0,0.0619857) (-7,0.0567185) +- (0,0.0753117) (-6,0.093921) +- (0,0.0969128) (-5,0.16243) +- (0,0.127448) (-4,0.285083) +- (0,0.168844) (-3,0.502764) +- (0,0.224224) (-2,0.888033) +- (0,0.297999) (-1,1.56931) +- (0,0.396146) (0,2.7737) +- (0,0.526659) (1,1.56931) +- (0,0.396146) (2,0.888033) +- (0,0.297999) (3,0.502764) +- (0,0.224224) (4,0.285083) +- (0,0.168844) (5,0.16243) +- (0,0.127448) (6,0.093921) +- (0,0.0969128) (7,0.0567185) +- (0,0.0753117) (8,0.0384222) +- (0,0.0619857) (9,0.0329333) +- (0,0.0573876)
};

\addplot [
color=red,
solid
]
plot [error bars/.cd, y dir = both, y explicit]
coordinates{ (-8,0.0384222) +- (0,0.0196016) (-7,0.0567185) +- (0,0.0238156) (-6,0.093921) +- (0,0.0306465) (-5,0.16243) +- (0,0.0403027) (-4,0.285083) +- (0,0.0533932) (-3,0.502764) +- (0,0.0709059) (-2,0.888033) +- (0,0.0942355) (-1,1.56931) +- (0,0.125272) (0,2.7737) +- (0,0.166544) (1,1.56931) +- (0,0.125272) (2,0.888033) +- (0,0.0942355) (3,0.502764) +- (0,0.0709059) (4,0.285083) +- (0,0.0533932) (5,0.16243) +- (0,0.0403027) (6,0.093921) +- (0,0.0306465) (7,0.0567185) +- (0,0.0238156) (8,0.0384222) +- (0,0.0196016) (9,0.0329333) +- (0,0.0181475)
};

\end{axis}
\end{tikzpicture}

\caption[Typical food distribution.]{Typical food distribution.  The mean number of food particles at each location relative to the injector as a function of the injection rate $\gamma$, for parameter values $\alpha=.1$, $y=1/3$. The error bars represent one standard deviation about the mean.  Note the vertical axis is scaled according to the injection rate $\gamma$.  As $\gamma$ increases the mean number of food particles a given distance from the source increases linearly in $\gamma$, while the standard deviation increases as $\gamma^{1/2}$, so the ratio of the standard deviation to the mean decreases.  
Hence the \emph{shape} of the plot is independent of the injection rate, but the \emph{number of standard deviations} between adjacent means increases without bound as $\gamma$ goes to infinity.  Increasing $\alpha$ or decreasing $y$ gradually narrows the distribution.  Subsequent results do not depend strongly on the choice of these parameters.}
   \label{fig:food-distribution}
\end{figure}

The means of the food distribution follow a strictly decreasing function of distance. Therefore the average amount of food is always higher nearer the injector and the expected amount food $m$ steps CW or CCW of the injector is the same. Additionally, the vector of means $\Lambda$ is directly proportional to the injection rate $\gamma$. Results for other food distributions are discussed in Section \ref{sec:results}.

A creature decides how to move based on its observations.  The random variable $C_t$ represents the creature's location at time $t$, and the random variable $Z_t$ represents the creature's observation at the location $C_t$ at time $t$. We consider two types of observations depending on the model, namely, molecule counts (in Model I) and distance to the source (in Model II).
Section \ref{sec:discrete_molecule} below describes these two scenarios in more detail.
In both cases, the current observation $Z_t$, given the creature's and source's locations at time $t$, is independent of past observations or locations.

We assume the creature is able to maintain a history of past locations and observations, for instance as a form of working memory. In Appendix \ref{app:histories} we introduce a formal definition of a history, parallel to the concept of a sigma algebra in the theory of stochastic processes, in order to obtain the results presented in Section \ref{sec:F}.  For ease of presentation here, however, it suffices to represent a history informally as a growing list of observations $Z_t=z_t$ and creature locations $C_t=c_t$.
%
%
Therefore we define a creature's \emph{history at time $t$} to be the list of observations and creature locations up to and including time $t$. A history is then $\{\{Z_i=z_i\},\{C_i=c_i\}\}_{i=0}^{t}$. \footnote{Unless otherwise stated, all references to histories are made with respect to time $t$.}
%
%
We define a creature's {\it strategy} to be a (possibly random) mapping from histories to a choice for the creature's next location. The strategies under consideration are discussed in Section \ref{sec:strategies_intro}.

The model creature will be endowed with perfect information about several aspects of its environment. In order to focus on the core issue of the (in)equivalence of information maximization and utility maximization, we will simplify the problem by assuming the
 creature
  has complete knowledge about every aspect of the world \emph{except} the injector's actual location. Therefore we assume the creature knows the following:
\begin{enumerate}
\item The creature knows that there are exactly $N$ different locations on a 1-dimensional, ring world.
\item The creature knows that at time $t=0$ the injector starts at a uniformly random location.
\item The creature knows the  history of its own locations and observations. \label{item:knowhistory}
\item The creature knows the injector's movement algorithm, $P_x$, and the value of $x$.
\item The creature knows that the molecules are distributed as Poisson random variables with known means  $\Lambda$ centered on the (unknown) source location.
\end{enumerate}
In addition, we assume the creature's movement is unconstrained, \textit{i.e.}~it can move to any position on the ring at the next iteration, without limitation.  This assumption greatly simplifies the formulation and analysis of strategies, as well shall see.  It is not entirely unreasonable, however, as in some instances animals may travel substantial  distances between foraging or search locations, a behavior known as \emph{saltatory search} or \emph{intermittent locomotion} \cite{KramerMcLaughlin2001AmZool,OBrienEvansBrowman:1989:Oecologia,ReynoldsFrye2007PLoSOne}.  Similarly, saltatory movements unconstrained by distance are observed in the pattern of saccades during visual search \cite{Najemnik+Geisler:2005:Nature}.

 
Assumption \ref{item:knowhistory} above would appear to require that the foraging creature retain an infinitely capacious memory of all of its prior locations and observations.  For the strategies under consideration, however, we will show in Section \ref{sec:F} that if the creature has the ability to maintain an internal state representing a single  probability mass function on $\ZN$, 
this knowledge is equivalent to the creature  knowing its full  history.  Hence only a finite dimensional ``memory" is required in order to exactly encapsulate the creature's entire history in a natural way.

\subsubsection{Discrete Molecule and Distance-Certain Models}
\label{sec:discrete_molecule}

In the Discrete Molecule model (Model I), the creature observes the number of molecules at its location at time $t$.  The Poisson distribution of food at each location relative to the injector is conditionally independent of past observations, given the current location of the source. Therefore given the source's and creature's locations, the creature's observation is independent of past observations.
The independence of successive observations and the creature's assumed knowledge of the world allows the creature to determine the probability mass function describing the injector's next location accurately (see Section \ref{sec:F}).  

\label{sec:distance_certain}

As the rate of food injection $\gamma$ increases, the mean number of molecule counts at each location grows linearly in $\gamma$ while the standard deviation about the mean grows as $\sqrt{\gamma}$.  Given an observation of molecule counts at time $t$, the creature can estimate the distance to the source because it knows the shape of the distribution of food relative to the source.  The food distribution is symmetric about the source; once released, food is presumed to move \textit{via} diffusion without drift (see Appendix \ref{app:food}). The creature can therefore only estimate the absolute distance, not the direction to the source, from the observed molecule count.  However, it can be argued \cite{Agarwala2009}, using  Chebyshev's inequality, that the probability of a creature mistaking the distance to the source approaches zero in the limit as $\gamma$ grows without bound.\footnote{Numerical simulations of this system show that moderately high injection rates lead to very low probabilities of error in estimating the distance to the source (not shown).} 
Consequently we introduce a second model, an idealization of Model I, in which the creature \emph{infallibly} determines the distance to the food source on each observation (although not necessarily which direction it lies in).  In this model the observation $Z_t$ is just $D_t$, the actual distance to the source.
 The \emph{Distance-Certain Model}, Model II, allows full analytical treatment of the asymptotic behavior of the creature under the different strategies of interest.

\subsection{Strategies}\label{sec:strategies_intro}
We focus on three search strategies.  The \emph{Information Maximization} (or \emph{Infomax}) strategy seeks to minimize the creature's uncertainty about the source location, as quantified by the entropy (see Appendix B) 
 of a probability mass function (PMF) representing the possible next source location, given the creatures' past observations.  If multiple locations satisfy this information-greedy algorithm, the creature randomly chooses between the locations with the highest expected food.\footnote{Alternatively the creature could chose the \textit{closest} location among those to which the algorithm is otherwise indifferent, rather than choosing several at random.}  

For purposes of comparison we introduce a second strategy, the \emph{Likelihood Maximization} (or \emph{Max-Likelihood}) strategy.  At each iteration, the creature following this strategy will move to the most likely next location of the source based on its prior observations.  As the food distribution peaks at the location of the source, this strategy is equivalent to a greedy or short-term maximization of expected gain.  (However, it does not necessarily maximize average expected gain over the long run.)  


In addition to the Infomax and Max-Likelihood strategies, we consider a third strategy combining elements of the first two.  Under the \emph{modified Maximum Likelihood} (or \emph{Hybrid}) strategy, the creature moves according to the Max-Likelihood strategy when it is \textit{certain} where the source is located, and otherwise it follows the Infomax strategy.  That is, the Hybrid creature attempts to minimize its uncertainty about the source until it actually locates it, after which it tries to colocalize with the source even at the risk of subsequently increasing its uncertainty.\footnote{This strategy could also be called \emph{track and pounce}.}  

We draw a distinction between \emph{survival strategies} (``maximize information''; ``maximize the likelihood of colocation'') and \emph{movement algorithms}, the latter being an implementation of a given strategy.\footnote{When there is no risk of ambiguity we may refer to strategies and the corresponding algorithms interchangeably.}
Appendix \ref{app:strategies} provides precise definitions of the algorithm corresponding to each of the three strategies described above.  The technical definitions of the Infomax and Max-Likelihood strategies will depend on the notion of a probability mass function encapsulating a creature's history (Section \ref{sec:F}); the definition of the Hybrid strategy will depend in addition on the definition of an information state (Section \ref{sec:infostates}).  These entities are introduced in Section \ref{sec:F}. 
While the algorithms are well defined for the full range of values of the source mobility parameter, $0\le x\le 1/2$, they only strictly correspond to implementations of the strategies for which they are named over a narrower range, $1/4<x<1/3$ (for further discussion of this point please see Section \ref{ssec:strat-vs-alg}.) We will focus our analysis on this narrower range of source mobilities.  


Although developed independently, our model system fits comfortably within the closed-loop information/control theoretic framework described in \cite{TouchetteLloyd:2000:PRL}.  The state variable subject to control ($X$, in \cite{TouchetteLloyd:2000:PRL}) is the displacement $S^t-C^t$ between the food source and the creature.  The ``noise'' introduced by the ``environment'' ($\mathcal{E}$) corresponds to the random walk performed by the food source, or $S^{t+1}-S^t$.  The measurement ($\mathcal{A}$)
is either the molecule count at location $C^t$ or else the observed distance $|S^t-C^t|$.  The control apparatus in their system ($\mathcal{C}$) corresponds to the strategy, along with the internal state of the creature, that together dictate the creature's next move.

Our model could also be seen as a special case of the general ``perception-action loop as a Bayesian network" discussed in (\cite{KlyubinPolaniNehaniv:2007:NeuralComp}, \textit{cf.}~their Figure 1).  In contrast with the particular model studied in \cite{KlyubinPolaniNehaniv:2007:NeuralComp}, our model allows an analytic, asymptotic analysis, rather than a computational analysis over a fixed number of time steps; and we consider the consequences of multiple behavioral strategies rather than information maximization alone.

\section{Results}\label{sec:results}
\subsection{Injector's Location Given a Creature's History}\label{sec:F}

The injector's location at time $t$, $S_t$, is a uniform random variable on the ring world (as shown in Lemma \ref{lem:probS}). However, the creature can use its history of observations to create a more informative distribution of the injector's current and (more importantly) next location.
We let $\pmfm{t}$ be the probability mass function representing the injector's location at time $t$ given the creature's prior history (the  history at time $t-1$).   
Furthermore, we let $\pmfp{t}$ be the random PMF of the injector's location at time $t$ given the creature's history at time $t$. 
(See Appendix \ref{app:histories} for formal definitions of histories and the PMFs $\pmfm{}$ and $\pmfp{}$.)  Note the PMFs $\pmfm{t}$ and $\pmfp{t}$ are themselves  random variables taking values in $\pmf$.
Given $\pmfp{t}$, it is easy to calculate $\pmfm{t+1}$ as the convolution of $\pmfp{t}$ with the injector's movement algorithm, $P$:
\begin{align}
 \pmfm{t+1}=&\pmfp{t}*P.\label{eq:pmfm}
\end{align}
Clearly $\pmfm{0}\equiv1/N$, the uniform distribution, due to the injector's starting at a random location with respect to the creature.\footnote{This statement and others are technically only true ``with probability one".  We dispense with this terminology except where necessary for clarity.}
Henceforth we will suppress the subscript $t$; all random variables will refer to the same time unless noted otherwise.

After a creature makes the observation $Z=z$,
its estimate of the source's location, represented by the PMF $\pmfp{}$, is updated in a Bayesian fashion (see Appendix \ref{app:bayes}).  The update is based on the conditional probability of making the observation $Z$ given the creature's location ($C$) and the possible source location $S=l$:
\begin{align}
 \pmfp{}(l)=&\frac{\Prob(Z=z|S=l\cap C=c) \pmfm{}(l)}{\sum_m \Prob(Z=z|S=m\cap C=c) \pmfm{}(m)}.\label{eq:pmfp}
\end{align}
The resulting expression, while somewhat unwieldy in the case of Model I, is nevertheless tractable because of the assumed Poisson nature of the molecule counts \cite{Agarwala2009}.
In the Distance-Certain Model, where the observation $Z$ is just the distance $D$, the right hand side of Equation \ref{eq:pmfp} simplifies considerably.  
The conditional probability of observing a given distance,
$\Prob(D=d|S=l\cap C=c)$,  equals one if the distance between $l$ and $c$ is equal to $d$, and zero otherwise.
Moreover, when $D$ is exactly equal to either zero or $N/2$, then Equation \ref{eq:pmfp} reduces to
\begin{align}
 \pmfp{}(l)=&\left\lbrace
\begin{array}{ll}
	1	& 	l=C+D\\
	0 	& 	\text{otherwise}.
\end{array}
\right. \label{eq:pmfpdc1}
\end{align}
And when $0>D>N/2$, Equation \ref{eq:pmfp} reduces to
\begin{align}
 \pmfp{}(l)=&\left\lbrace
\begin{array}{ll}
	\frac{\pmfm{}(C+D)}{\pmfm{}(C+D)+\pmfm{}(C-D)} 	& 	l=C+D\\
	\frac{\pmfm{}(C-D)}{\pmfm{}(C+D)+\pmfm{}(C-D)}	& 	l=C-D\\
	0 	& 	\text{otherwise.}
\end{array}
\right.
\label{eq:pmfpdc2}
\end{align}
Equation \ref{eq:pmfpdc1} differs from Equation \ref{eq:pmfpdc2} because when $D$ is exactly 0 or $N/2$, then $C+D= C-D,\mod N$. Figure 2.1 
shows examples of probability mass functions arising in the Distance-Certain Model.

It is biologically implausible to assume that a creature could retain an unlimited list of all of its past locations and observations in order to implement a search strategy.  Fortunately, Equations \ref{eq:pmfm} and \ref{eq:pmfp} demonstrate that with regard to the injector's location, a  lifetime of observations and locations can be captured \emph{exactly} by an appropriate probability mass function. What is more, in both the Discrete Molecule and Distance-Certain models, the PMF required can be updated iteratively in a lossless fashion to incorporate the creature's latest observation and location.
%
%
Finally, because the PMF of the injector's next location (conditioned on a history) can be calculated in a lossless fashion, there is no theoretical limit on the ability of past observations to accurately inform future predictions.

\subsection{Coarse Graining: Information States}\label{sec:infostates}
In Section \ref{sec:strategies_intro} a movement algorithm is defined as a (possibly random) mapping from creature histories to the creature's next location. We now define a {\it Markovian movement algorithm} to be any mapping of a creature's probability mass function $\pmfm{t+1}$ (describing the injector's \emph{next} location, given the creature's history at time $t$) to the creature's next location $C_{t+1}$.  
As with a movement algorithm, a \emph{Markovian} movement algorithm may be  random  or determinstic.
Any Markovian algorithm defines a random process on the state comprising the source, creature, observations, and the PMFs $\pmfm{}$ and $\pmfp{}$.  By construction, this random walk satisfies the Markov property, i.e.~given its state at time $t$, its state at time $t+1$ is conditionally independent of its states at all preceding times $t'<t$. The algorithm corresponding to each of the strategies to be considered (Infomax, Max-Likelihood, and Hybrid) is Markovian, as will become evident in Section \ref{sec:formal_strategies}.

In both the Discrete Molecule and Distance-Certain models, the random walk takes place on an infinite state space.  For the Distance-Certain model, however, it is possible to define  a finite number of coarse-grained \emph{information states} with respect to which the system constitutes a finite dimensional Markov process.
In most instances, coarse graining a Markov process on a network produces a random process on a smaller network that no longer satisfies the Markov property. 
However, for each of the Markov strategies considered,
the coarse graining we construct maps the random walk on the original  ``microscopic" states of the Distance-Certain model to information states in such a way as to preserve the Markovian character of the random walk (as we will show in this Section).  Consequently 
we are able to apply standard results from the theory of Markov processes to obtain analytical results that in turn shed light on the biological merits of the different strategies.\footnote{In the Markov chain literature, an exact coarse-graining is also called a \emph{lumping} of a Markov process.}  

%

The information states not only provide a coarse graining that preserves the Markov property for the system, they also correspond to intuitively appealing equivalence classes that capture biologically salient conditions such as the actual distance between creature and source and the creature's knowledge about the source.  The classification of information states turns on the structure of the probability mass function $\pmfp{}$ reflecting the likelihood that the source is at a given position relative to the creature.   We define  
$\ISp{}$ to be the set of PMFs for which the injector's location is known, or is known to be some definite distance $d>0$ from the creature's location. In the Distance-Certain Model the probability that $\pmfp{}$ is in $\ISp{}$ is equal to one, by construction.
The set $\ISp{}$ can be partitioned into the sets $\ISp{a}$ with $a\in\{0,1,2,3\}$ and the set $\ISp{I}$. For $a\in\{0,1,2\}$ the set $\ISp{a}$ is the set of PMFs representing the condition that there is some location, $e$, such that the injector is equally like to be at the location(s) $e\pm a$ and has zero probability of being else where. The set $\ISp{3}$ is the set PMFs where the injector is equally likely to be at the locations $e\pm d$ with $d\ge3$ and has zero probability of being elsewhere.
Figure 2.1 
illustrates typical PMFs of these states.
The set $\ISp{I}$ contains the remaining PMFs in $\ISp{}$, for instance those for  which it is certain the injector is a given distance from some location $e$, but the two possible locations do not have equal probability.  
\begin{figure}
\label{fig:information_states+}
\begin{minipage}[htp]{0.5\linewidth} 
\centering
\begin{tikzpicture}[line width=1pt]
\foreach \xcv/\ymv in 
{{2.2cm}/{2cm}}
{
\path
    coordinate (xc) at (\xcv,0)
    coordinate (ym) at (0,\ymv);

 \draw (0,0) -- ($2*(xc)$);
\draw (0,0) -- ($(ym)+(0,.2)$);

\foreach \ytick / \ylabel in
{{(ym)}/1, {($1/2*(ym)$)}/{$\frac{1}{2}$}}
{
\begin{scope}[shift=\ytick]
\draw  (.1,0) -- (-.1,0);
\node [label=180:\ylabel] at (0,0) {};
\end{scope}
};

\begin{scope}[shift={($(xc)$)}]
\foreach \xdiff / \xlabel /\px in
{0/$e$/1}
{
\begin{scope} [shift={($(\xdiff,0)$)}]
\draw ($(0,.1)$) -- ($(0,-.1)$); 
\node at (0,-.4) {\xlabel}; 
\draw[ycomb,
color=gray,
line width=0.2cm]
(0,0) -- (0,\px*\ymv); 
\end{scope};
}
\end{scope};

\node at ($(xc)+(0,-.9)$) {{Location $(l)$}};
\node[rotate=90] at ($(-1,0)+1/2*(ym)$) {{$f^+(l)$}};

\node at ($(xc)+(ym)+(0,.4)$) {{\bf $\pmfp{}\in\text{ }\ISp{0}$}};
}
\end{tikzpicture}

\hspace{5mm}

\begin{tikzpicture}[line width=1pt]
\foreach \xcv/\ymv in 
{{2.2cm}/{2cm}}
{
\path
    coordinate (xc) at (\xcv,0)
    coordinate (ym) at (0,\ymv);

 \draw (0,0) -- ($2*(xc)$);
\draw (0,0) -- ($(ym)+(0,.2)$);

\foreach \ytick / \ylabel in
{{(ym)}/1, {($1/2*(ym)$)}/{$\frac{1}{2}$}}
{
\begin{scope}[shift=\ytick]
\draw  (.1,0) -- (-.1,0);
\node [label=180:\ylabel] at (0,0) {};
\end{scope}
};

\begin{scope}[shift={($(xc)$)}]
\foreach \xdiff / \xlabel /\px in
{0/$e$/0, -1./$e-2$/.5, 1./$e+2$/.5}
{
\begin{scope} [shift={($(\xdiff,0)$)}]
\draw ($(0,.1)$) -- ($(0,-.1)$); 
\node at (0,-.4) {\xlabel}; 
\draw[ycomb,
color=gray,
line width=0.2cm]
(0,0) -- (0,\px*\ymv); 
\end{scope};
}
\end{scope};

\node at ($(xc)+(0,-.9)$) {{Location $(l)$}};
\node[rotate=90] at ($(-1,0)+1/2*(ym)$) {{$f^+(l)$}};

\node at ($(xc)+(ym)+(0,.4)$) {{\bf $\pmfp{}\in\text{ }\ISp{2}$}};
}
\end{tikzpicture}

\end{minipage}
\begin{minipage}[htp]{0.5\linewidth}

\centering

\begin{tikzpicture}[line width=1pt]
\foreach \xcv/\ymv in 
{{2.2cm}/{2cm}}
{
\path
    coordinate (xc) at (\xcv,0)
    coordinate (ym) at (0,\ymv);

 \draw (0,0) -- ($2*(xc)$);
\draw (0,0) -- ($(ym)+(0,.2)$);

\foreach \ytick / \ylabel in
{{(ym)}/1, {($1/2*(ym)$)}/{$\frac{1}{2}$}}
{
\begin{scope}[shift=\ytick]
\draw  (.1,0) -- (-.1,0);
\node [label=180:\ylabel] at (0,0) {};
\end{scope}
};

\begin{scope}[shift={($(xc)$)}]
\foreach \xdiff / \xlabel /\px in
{0/{}/0, -.5/$e-1$/.5, .5/$e+1$/.5}
{
\begin{scope} [shift={($(\xdiff,0)$)}]
\draw ($(0,.1)$) -- ($(0,-.1)$); 
\node at (0,-.4) {\xlabel}; 
\draw[ycomb,
color=gray,
line width=0.2cm]
(0,0) -- (0,\px*\ymv); 
\end{scope};
}
\end{scope};

\node at ($(xc)+(0,-.9)$) {{Location $(l)$}};
\node[rotate=90] at ($(-1,0)+1/2*(ym)$) {{$f^+(l)$}};

\node at ($(xc)+(ym)+(0,.4)$) {{\bf $\pmfp{}\in\text{ }\ISp{1}$}};
}
\end{tikzpicture}

\hspace{5mm}

\begin{tikzpicture}[line width=1pt]
\foreach \xcv/\ymv in 
{{2.2cm}/{2cm}}
{
\path
    coordinate (xc) at (\xcv,0)
    coordinate (ym) at (0,\ymv);

 \draw (0,0) -- ($2*(xc)$);
\draw (0,0) -- ($(ym)+(0,.2)$);

\foreach \ytick / \ylabel in
{{(ym)}/1, {($1/2*(ym)$)}/{$\frac{1}{2}$}}
{
\begin{scope}[shift=\ytick]
\draw  (.1,0) -- (-.1,0);
\node [label=180:\ylabel] at (0,0) {};
\end{scope}
};

\begin{scope}[shift={($(xc)$)}]
\foreach \xdiff / \xlabel /\px in
{0/$e$/0, -1.5/$e-d$/.5, 1.5/$e+d$/.5}
{
\begin{scope} [shift={($(\xdiff,0)$)}]
\draw ($(0,.1)$) -- ($(0,-.1)$); 
\node at (0,-.4) {\xlabel}; 
\draw[ycomb,
color=gray,
line width=0.2cm]
(0,0) -- (0,\px*\ymv); 
\end{scope};
}
\end{scope};

\node at ($(xc)+(0,-.9)$) {{Location $(l)$}};
\node[rotate=90] at ($(-1,0)+1/2*(ym)$) {{$f^+(l)$}};

\node at ($(xc)+(ym)+(0,.4)$) {{\bf $\pmfp{}\in\text{ }\ISp{3}$}};
}
\end{tikzpicture}
\end{minipage}
\caption[PMFs in $\ISp{}$ for the Distance-Certain Model.]{The primary information states for the set of probability mass functions $\ISp{}$, representing the conditional probability of the current source location given a creature's history of observations through the present time. The function $f^+$ is a probability mass function on the ring world $\ZN$. For every $f^+$ in $\ISp{0}$ there is some unique location $0\le e\le N-1$ such that $f^+(e)=1$ and $f^+(l)=0$ for all other locations $l\ne e$. For the other information states there are two locations equidistant from some location $e$ where the injector is equally likely to be; furthermore the probability of the injector being at any other location is zero.  That is, $f^+(e-d)=f^+(e+d)=1/2$, and $f^+(l)=0$ for $l\ne e\pm d$.  For $f^+$ to be in $\ISp{1}$ or $\ISp{2}$, $d$ must equal one or two, respectively.  For $f^+$ to be in $\ISp{3}$, $d$ must be greater than two but less than or equal to $N/4$.}
\end{figure}
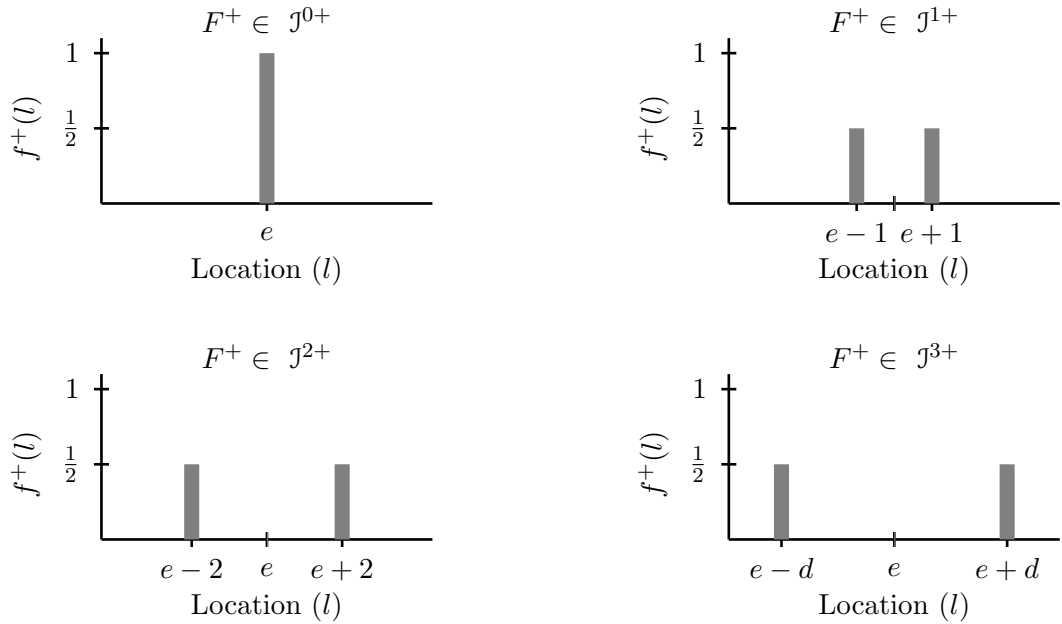

For each probability mass function $f^+\in\ISp{}$ representing the likelihood of the source's current location, there is a corresponding PMF representing the likelihood of its next location.  We define $\ISm{}$ as the set of PMFs that can be generated by Equation \ref{eq:pmfm} given $\pmfp{}\in\ISp{}$; for completeness we also include in $\ISm{}$  the uniform PMF to account for the initial state of ignorance about the source, $\pmfm{0}$. We partion the set $\ISm{}$ into subsets $\ISm{a}$ with $a\in\{0,1,2,3,I\}$ and $\ISm{*}=\{f_{u}\}$ (the uniform distribution, $f_u\equiv1/N$). Except for the set $\ISm{*}$, a PMF $f^-$ is in the set $\ISm{a}$ if and only if there is a PMF $f^+$ in $\ISp{a}$ such that $f^-=f^+*P_x$.
Figure \ref{fig:information_states-} shows typical PMFs of these states. Each such partition corresponds to an information state; we say a creature is \emph{in the information state} $\IS^{a\pm}$ when the random PMF $F^\pm$ takes a value $f^\pm$ in the state $\IS^{a\pm}$.



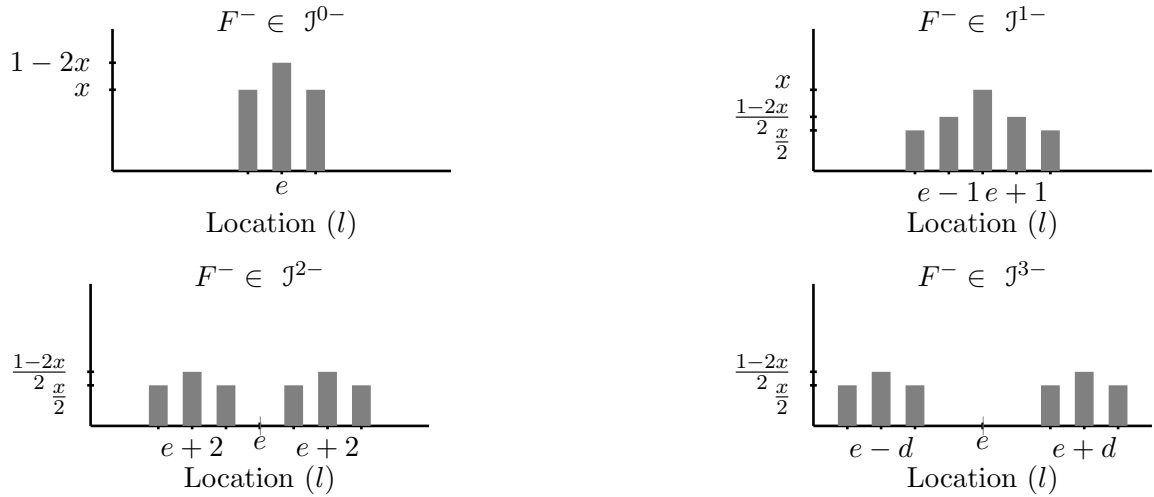
\begin{figure}
\label{fig:information_states-}
\begin{multicols}{2}
\begin{tikzpicture}[line width=1pt,scale=.45]
\foreach \xcv/\ymv in 
{{5cm}/{8cm}}
{
\path
    coordinate (xc) at (\xcv,0)
    coordinate (ym) at (0,\ymv);

 \draw (0,0) -- ($2*(xc)$);
\draw (0,0) -- ($1/2*(ym)+(0,.2)$);

\foreach \ytick / \ylabel in
{{($.3*(ym)$)}/{$x$}, {($.4*(ym)$)}/{$1-2x$}}
{
\begin{scope}[shift=\ytick]
\draw  (.1,0) -- (-.1,0);
\node [label=180:\ylabel] at (0,0) {};
\end{scope}
};

\begin{scope}[shift={($(xc)$)}]
\foreach \xdiff / \xlabel /\px in
{0/$e$/.4, -1/{}/.3, 1/{}/.3}
{
\begin{scope} [shift={($(\xdiff,0)$)}]
\draw ($(0,.1)$) -- ($(0,-.1)$); 
\node[anchor=north] at (0,0) {\xlabel}; 
\draw[ycomb,
color=gray,
line width=0.25cm]
(0,0) -- (0,\px*\ymv); 
\end{scope};
}
\end{scope};

\node at ($(xc)+(0,-1.6)$) {{Location $(l)$}};

\node at ($(xc)+1/2*(ym)+(0,.4)$) {{\bf $\pmfm{}\in\text{ }\ISm{0}$}};
}
\end{tikzpicture}

\begin{tikzpicture}[line width=1pt,scale=.45]
\foreach \xcv/\ymv in 
{{5cm}/{8cm}}
{
\path
    coordinate (xc) at (\xcv,0)
    coordinate (ym) at (0,\ymv);

 \draw (0,0) -- ($2*(xc)$);
\draw (0,0) -- ($1/2*(ym)+(0,.2)$);

\foreach \ytick / \ylabel in
{{($.15*(ym)$)}/{}, {($.2*(ym)$)}/{$\frac{1-2x}{2}$}}
{
\begin{scope}[shift=\ytick]
\draw  (.1,0) -- (-.1,0);
\node [label=180:\ylabel] at (0,0) {};
\end{scope}
};

\node [label=180:$\frac{x}{2}$] at ($.1*(ym)$) {};

\begin{scope}[shift={($(xc)$)}]
\foreach \xdiff / \xlabel /\px in
{0/$e$/0, -3/{}/.15, 3/{}/.15, -1/{}/.15, 1/{}/.15, -2/{$e+2$}/.2, 2/{$e+2$}/.2}
{
\begin{scope} [shift={($(\xdiff,0)$)}]
\draw ($(0,.1)$) -- ($(0,-.1)$); 
\node[anchor=north] at (0,0) {\xlabel}; 
\draw[ycomb,
color=gray,
line width=0.25cm]
(0,0) -- (0,\px*\ymv); 
\end{scope};
}
\end{scope};

\node at ($(xc)+(0,-1.6)$) {{Location $(l)$}};

\node at ($(xc)+1/2*(ym)+(0,.4)$) {{\bf $\pmfm{}\in\text{ }\ISm{2}$}};
}
\end{tikzpicture}

\centering

\begin{tikzpicture}[line width=1pt,scale=.45]
\foreach \xcv/\ymv in 
{{5cm}/{8cm}}
{
\path
    coordinate (xc) at (\xcv,0)
    coordinate (ym) at (0,\ymv);

 \draw (0,0) -- ($2*(xc)$);
\draw (0,0) -- ($1/2*(ym)+(0,.2)$);

\foreach \ytick / \ylabel in
{{($.15*(ym)$)}/{}, {($.2*(ym)$)}/{$\frac{1-2x}{2}$}, {($.3*(ym)$)}/{}}
{
\begin{scope}[shift=\ytick]
\draw  (.1,0) -- (-.1,0);
\node [label=180:\ylabel] at (0,0) {};
\end{scope}
};

\node [label=180:$\frac{x}{2}$] at ($.1*(ym)$) {};
\node [label=180:$x$] at ($.33*(ym)$) {};

\begin{scope}[shift={($(xc)$)}]
\foreach \xdiff / \xlabel /\px in
{0/{}/.3, -2/{}/.15, 2/{}/.15, -1/$e-1$/.2, 1/$e+1$/.2}
{
\begin{scope} [shift={($(\xdiff,0)$)}]
\draw ($(0,.1)$) -- ($(0,-.1)$); 
\node[anchor=north] at (0,0) {\xlabel}; 
\draw[ycomb,
color=gray,
line width=0.25cm]
(0,0) -- (0,\px*\ymv); 
\end{scope};
}
\end{scope};

\node at ($(xc)+(0,-1.6)$) {{Location $(l)$}};

\node at ($(xc)+1/2*(ym)+(0,.4)$) {{\bf $\pmfm{}\in\text{ }\ISm{1}$}};
}
\end{tikzpicture}

\begin{tikzpicture}[line width=1pt,scale=.45]
\foreach \xcv/\ymv in 
{{5cm}/{8cm}}
{
\path
    coordinate (xc) at (\xcv,0)
    coordinate (ym) at (0,\ymv);

 \draw (0,0) -- ($2*(xc)$);
\draw (0,0) -- ($1/2*(ym)+(0,.2)$);

\foreach \ytick / \ylabel in
{{($.15*(ym)$)}/{}, {($.2*(ym)$)}/{$\frac{1-2x}{2}$}}
{
\begin{scope}[shift=\ytick]
\draw  (.1,0) -- (-.1,0);
\node [label=180:\ylabel] at (0,0) {};
\end{scope}
};

\node [label=180:$\frac{x}{2}$] at ($.1*(ym)$) {};

\begin{scope}[shift={($(xc)$)}]
\foreach \xdiff / \xlabel /\px in
{0/$e$/0, -4/{}/.15, 4/{}/.15, -2/{}/.15, 2/{}/.15, -3/{$e-d$}/.2, 3/{$e+d$}/.2}
{
\begin{scope} [shift={($(\xdiff,0)$)}]
\draw ($(0,.1)$) -- ($(0,-.1)$); 
\node[anchor=north] at (0,0) {\xlabel}; 
\draw[ycomb,
color=gray,
line width=0.25cm]
(0,0) -- (0,\px*\ymv); 
\end{scope};
}
\end{scope};

\node at ($(xc)+(0,-1.6)$) {{Location $(l)$}};

\node at ($(xc)+1/2*(ym)+(0,.4)$) {{\bf $\pmfm{}\in\text{ }\ISm{3}$}};
}
\end{tikzpicture}

\end{multicols}

 \caption[Information States following Injector Movement.]{Information states following injector movement, and before the next observation.  The function $f^-$ is a probability mass function on the ring world $\ZN$, representing the likelihood of finding the food source at a given location \emph{after} its subsequent move.  Each $f^-$ in  $\ISm{}$ is obtained by convolving a corresponding PMF $f^+$ in $\ISp{}$ with the transition matrix kernel $[\cdots, x,(1-2x),x,\cdots]$ for the source movement generator $P_x$.  Compare Figure 2.1. 
 Note the shapes of the PMFs depicted are qualitatively correct provided we restrict the range of the source mobility parameter $x$.  For instance a PMF in $\ISm{1}$ takes the values $f^-(e)=x$, $f^-(e\pm 1)=(1-2x)/2$, and $f^-(e\pm 2)=x/2$ for some location $e$ in the ring world, $0\le e \le N-1$.  The PMF has the shape shown (increasing monotonically to a central peak at location $e$) provided $x/2<(1-2x)/2<x$, which is equivalent to assuming $1/4<x<1/3$.  For this figure we set $x=0.3$.}

\end{figure}

\subsection{Markov Chains on Information States}
\label{sec:markov}
When the creature is in the information state $\ISm{0}$ (see Figure \ref{fig:information_states-}) it has a PMF of the form
\begin{align}\label{eq:f-0}
	f^-(l)&=
\left\lbrace
\begin{array}{ll}
	x 	& 	l=e\pm1\\
	1-2x	& 	l=e\\
	0 	& 	\text{otherwise.}
\end{array}
\right.
\end{align}
Because in this state the current location of the source is known with certainty, the location of the source at the next time step is known to obey the distribution given in Equation \ref{eq:f-0}.  Under the Infomax strategy, the creature in information state $\ISm{0}$ must chose to move to location $e+1$ (or equivalently, $e-1$) because this move ensures that the creature will know the state with certainty again after the next observation.  To see this one need only consider three cases.  Given that $S_t=e$ is known, if $C_{t+1}=e-1$ is chosen then the displacement, $S_{t+1}-C_{t+1}$, can only take one of three values: 0 if the source moves CCW by one step; 1 if the source does not move; 2 if the source moves CW by one step.  The Max-Likelihood strategy, in contrast, calls for setting $C_{t+1}=e$, i.e.~placing the creature at the next time step at the most likely location of the source.  However, with probability $2x$, the source will have moved either CW and CCW, and the creature's observation of the molecule count or the distance, taken at $C_{t+1}$ will leave the new location of the source ambiguous.  
In this case the goals of minimizing uncertainty and maximizing immediate food gain lead to two different choices.  Similar calculations lead to a set of (probabilistic) transitions between the information states for each of the three strategies considered.  We can capture the transitions induced by each strategy succinctly in a matrix representation that makes the Markov chain structure of the coarse grained system explicit.  Let $p_{t}$ be the (row) vector representing the probability of being in the initial state ($\ISm{*}$), or state $\ISm{a}$, $a=0,1,2,3$, respectively, on iteration $t$.   Then for each strategy we have a $5\times 5$ matrix $M$ such that $p_{t+1}=p_t M$.   Table \ref{tab:5x5matrices} shows the resulting transition matrices for the three strategies considered.  

\begin{table}[htbp]
   \centering
      \begin{tabular}{lc} 
     Infomax & $M_{ITC}=
\left(\begin{array}{ccccc}
     0&\frac{2}{N}&\frac{4}{N}&\frac{4}{N}&\frac{N-10}{N} \\ 
\cline{2-2}
     0& \multicolumn{1}{|c|}{\mathbf{1}}&0&0&0\\
\cline{2-2}
     0&1&0&0&0\\
     0&1&0&0&0\\
     0&1&0&0&0
      \end{array} \right)$\\ \hline \\ 
     Max-Likelihood & $M_{MLC}=\left(\begin{array}{ccccc}
     0&\frac{2}{N}&\frac{4}{N}&\frac{4}{N}&\frac{N-10}{N} \\  
\cline{2-4}
     0&\multicolumn{1}{|c}{\mathbf{(1-2x)}}&\mathbf{2x}&\multicolumn{1}{c|}{\mathbf{0}}&0\\
     0&\multicolumn{1}{|c}{\mathbf{x}}&\mathbf{(1-2x)}&\multicolumn{1}{c|}{\mathbf{x}}&0\\
     0&\multicolumn{1}{|c}{\mathbf{(1-x)}}&\mathbf{x}&\multicolumn{1}{c|}{\mathbf{0}}&0\\
\cline{2-4}
     0&(1-x)&x&0&0
      \end{array} \right)$\\ \hline \\
     Hybrid & $M_{mMLC}=\left(\begin{array}{ccccc}
     0&\frac{2}{N}&\frac{4}{N}&\frac{4}{N}&\frac{N-10}{N} \\ 
\cline{2-3}
     0&\multicolumn{1}{|c}{\mathbf{(1-2x)}}&\multicolumn{1}{c|}{\mathbf{2x}}&0&0\\
     0&\multicolumn{1}{|c}{\mathbf{1}}&\multicolumn{1}{c|}{\mathbf{0}}&0&0\\
\cline{2-3}
     0&1&0&0&0\\
     0&1&0&0&0
      \end{array} \right)$
   \end{tabular}
   \caption[Transition matrices for each strategy.]{Transition matrices for each strategy. TOP: Information Theory Creature (ITC).  CENTER: Maximum Likelihood Creature (MLC).  BOTTOM: modified Maximum Likelihood Creature, or Hybrid (mMLC).  Each row represents the probability, under the given strategy, of moving from one state (in order: $\ISm{*}, \ISm{0},\ISm{1},\ISm{2},\ISm{3}$) to any of the same five states. Each row sums to unity.  For $M_{ITC}$, there is a minimal absorbing set comprising a single state; for $M_{MLC}$ there is a minimal absorbing set of three states, and for $M_{mMLC}$ (the hybrid strategy) there are two. In each case the absorbing set is indicated by the boxed terms.}
   \label{tab:5x5matrices}
\end{table}

The initial information state is necessarily the maximally uninformed state $\ISm{*}$ with initial PMF $F^-\equiv 1/N$, regardless of the creature's strategy.  
After the first measurement the creature enters one of the other states $\ISm{a}, 0\le a \le 3$.  
For each transition matrix $M$ there is a subset of information states comprising a minimal absorbing set.  It is straightforward to show that for each strategy considered, the Markov chain when restricted to the minimal absorbing set is positive recurrent (see Appendix \ref{app:strategytransitionmatrices}).  Therefore by the Perron-Frobenius theorem there is a unique stationary distribution $\pi\in\pmf$ such that $\pi M=\pi$ and $p_t\to\pi$ as $t\to \infty$.  The stationary distribution may be calculated explicitly for each strategy and for each value of the source mobility, $x$.  The properties of this asymptotic distribution of state occupancies in turn allow us to calculate the average amount of food obtained in the long run for each strategy, as well as the mean uncertainty (entropy) of the source position at each step, given the creature's prior observations.

As an example we will present the calculation of the expected food gained for the simplest case.  Suppose, as above, that the creature is in information state $\ISm{0}$, i.e.~at time $t$ the creature knows with certainty that the source is currently at some location $S_t=e$.  Recall that the food concentration (Model II) or mean food count (Model I) $m$ steps CW from the source is given by $\lambda_m, 0\le m \le N-1$.  
If the creature chooses as its next location $C_{t+1}=e+1$ (or equivalently $e-1$), then there is a probability $x$ that the creature lands on the source and sees (on average) a quantity $\lambda_0$ of food. With probability $1-2x$, the source will be at location $e$ (one step away) and so the creature will see an average food count of $\lambda_1$. Finally the creature will see on average $\lambda_2$ when the source is at $e-1$ with probability $x$. Therefore the expected food for location $e+1$ when in the state $\ISm{0}$ is $x\lambda_0+(1-2x)\lambda_1+x\lambda_2$.
The other cases proceed similarly; 
Figures \ref{fig:is0-transitions}, \ref{fig:is1-transitions}, \ref{fig:is2-transitions}, and \ref{fig:is3-transitions} present the results on transitions between information states and expected food gains for each recurrent information state. 



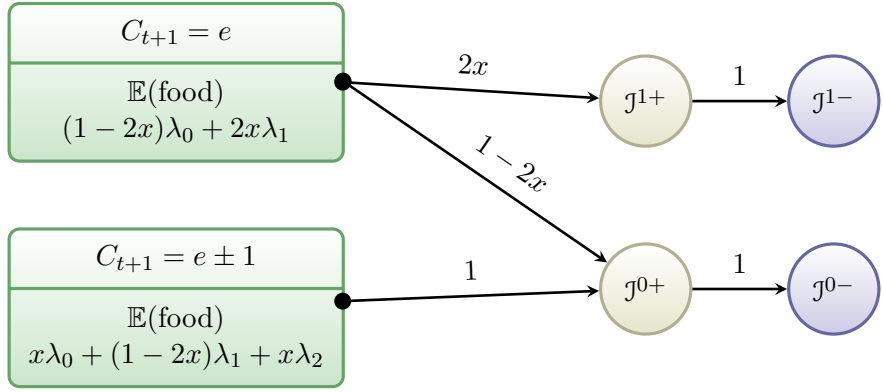
\begin{figure}
\begin{center}
\begin{tikzpicture}
[inner sep=2mm,node distance=2.5cm,auto, line width=1pt,
statep/.style={circle, shade, top color=white,
    bottom color=yellow!50!black!20, draw=yellow!40!black!60, very
    thick,minimum size=8mm},
statem/.style={circle, shade, top color=white,
    bottom color=blue!50!black!20, draw=blue!40!black!60, very
    thick,minimum size=8mm},
 location/.style={rectangle, rounded corners, shade, top color=white,
    bottom color=green!50!black!20, draw=green!40!black!60, very
    thick,text width=4cm,text badly centered }
]

\linespread{1}
\node (e) at (0,0) [anchor=east, location, rectangle split, rectangle split parts=2] {$C_{t+1}=e$ \nodepart{second} $\E(\text{food})$\\$(1-2x)\lambda_0+2x\lambda_1$};
\node (epm1) at (0,-3) [anchor=east, location, rectangle split, rectangle split parts=2] {$C_{t+1}=e\pm 1$ \nodepart{second} $\E(\text{food})$\\$x\lambda_0+(1-2x)\lambda_1+x\lambda_2$};
\linespread{1.6}

\node (IS1p) at  (4,-.25) [statep] {$\ISp{1}$};
\node (IS0p) [below of=IS1p] [statep] {$\ISp{0}$};
\node (IS0m) [right of=IS0p] [statem]  {$\ISm{0}$};
\node (IS1m) [right of=IS1p] [statem] {$\ISm{1}$};

\path[>=stealth,*->,shorten <=-4pt] 
	  (epm1)		edge	[sloped]		node		{1}		(IS0p)
	  (e.east)	edge	[sloped,pos =.4]	node		{$1-2x$}	(IS0p)
			edge	[sloped,above]		node		{$2x$}		(IS1p)
;

\path[>=stealth,->] 
(IS0p) 	edge	[sloped]		node		{1}		(IS0m)
	  (IS1p)	edge	[sloped]		node		{1}		(IS1m);
\end{tikzpicture}
\caption[Transitions from $\ISm{0}$.]{Transitions from $\ISm{0}$. This figure states the expected food at each choice for the creature's next location ($C_{t+1}$) when in the information state $\ISm{0}$. It also depicts the transition probabilities to other information states based on $C_{t+1}$. By definition, if the creature is in $\ISp{0}$ or $\ISp{1}$ at time $t$ the creature will be in $\ISm{0}$ or $\ISm{1}$ 
at time $t+1$ (respectively). Future graphs will not depict the intermediate information states.}\label{fig:is0-transitions}
\end{center}
\end{figure}

\begin{figure}
 \begin{center}
\begin{tikzpicture}
[inner sep=2mm,node distance=2.5cm,auto, line width=1pt, shorten <=-4pt,
statep/.style={circle, shade, top color=white,
    bottom color=yellow!50!black!20, draw=yellow!40!black!60, very
    thick,minimum size=8mm},
statem/.style={circle, shade, top color=white,
    bottom color=blue!50!black!20, draw=blue!40!black!60, very
    thick,minimum size=8mm},
 location/.style={rectangle, rounded corners, shade, top color=white,
    bottom color=green!50!black!20, draw=green!40!black!60, very
    thick,text width=4cm,text badly centered },
crossline/.style={preaction={draw=white, -,shorten <=10pt,
line width=6pt}}]

\linespread{1}
\node (e) at (0,0) [anchor=east, location, rectangle split, rectangle split parts=2] {$C_{t+1}=e$ \nodepart{second} $\E(\text{food})$\\$x\lambda_0+(1-2x)\lambda_1+x\lambda_2$};

\node (epm1) at (0,-3) [anchor=east, location, rectangle split, rectangle split parts=2] {$C_{t+1}=e\pm 1$ \nodepart{second} $\E(\text{food})$\\$\frac{(1-2x)(\lambda_0+\lambda_2)}{2}+\frac{x(3\lambda_1+\lambda_3)}{2}$};

\node (epm2) at (0,-6) [anchor=east, location, rectangle split, rectangle split parts=2] {$C_{t+1}=e\pm2$ \nodepart{second} $\E(\text{food})$\\$\frac{x(\lambda_0+2\lambda_2+\lambda_4)}{2}+\frac{(1-2x)(\lambda_1+\lambda_3)}{2}$};
\linespread{1.6}

\node (ISIm) at  (4,.75) [statem]  {$\ISm{I}$};
\node (IS2m) [below of=ISIm] [statem]  {$\ISm{2}$};
\node (IS1m) [below of=IS2m] [statem] {$\ISm{1}$};
\node (IS0m) [below of=IS1m] [statem] {$\ISm{0}$};

\path[>=latex,*->] (epm2.east)	edge	[sloped]			node		{1}			(IS0m)
	  (e.east)	edge	[sloped, pos=.2,below]		node		{$x$}			(IS0m)
			edge	[sloped]			node		{$1-2x$}		(IS1m)
			edge	[sloped, pos=.35,above]		node		{$x$}			(IS2m)
	  (epm1.east)	edge	[crossline]								(ISIm)
			edge	[sloped,below,pos=.2]		node		{$\frac{3x}{2}$}	(ISIm)
			edge	[sloped,pos =.4,below]		node		{$1-\frac{3x}{2}$}	(IS0m)
			
;

\end{tikzpicture}
\caption[Transitions from $\ISm{1}$.]{Transitions from $\ISm{1}$. This figure states the expected food at each choice for the creature's next location ($C_{t+1}$) when in the information state $\ISm{1}$. It also depicts the transition probabilities to other information states based on $C_{t+1}$.}\label{fig:is1-transitions}
\end{center}
\end{figure}

\begin{figure} 
\begin{center}
\begin{tikzpicture}
[inner sep=2mm,node distance=2.5cm,auto, line width=1pt, shorten <=-4pt,
statep/.style={circle, shade, top color=white,
    bottom color=yellow!50!black!20, draw=yellow!40!black!60, very
    thick,minimum size=8mm},
statem/.style={circle, shade, top color=white,
    bottom color=blue!50!black!20, draw=blue!40!black!60, very
    thick,minimum size=8mm},
 location/.style={rectangle, rounded corners, shade, top color=white,
    bottom color=green!50!black!20, draw=green!40!black!60, very
    thick,text width=4cm,text badly centered },
crossline/.style={preaction={draw=white, -,shorten <=10pt,
line width=6pt}}]
\linespread{1}
\node (epm1) at (0,.5) [anchor=east, location, rectangle split, rectangle split parts=2] {$C_{t+1}=e\pm1$\nodepart{second} $\E(\text{food})$\\$ \frac{(1-2x)(\lambda_1+\lambda_3)}{2}+\frac{x(\lambda_0+2\lambda_2+\lambda_4)}{2}$};

\node (epm2) at (0,-2.5) [anchor=east, location, rectangle split, rectangle split parts=2] {$C_{t+1}=e\pm 2$\nodepart{second} $\E(\text{food})$\\$\frac{(1-2x)(\lambda_0+\lambda_4)}{2}+\frac{x(2\lambda_1+\lambda_3+\lambda_5)}{2} $};

\node (epm3) at (0,-5.5) [anchor=east, location, rectangle split, rectangle split parts=2] {$C_{t+1}=e\pm3$\nodepart{second} $\E(\text{food})$\\$\frac{x(\lambda_0+\lambda_2+\lambda_4+\lambda_6)}{2}+\frac{(1-2x)(\lambda_1+\lambda_5)}{2} $};
\linespread{1.6}

\node (IS2m) at  (4,0) [statem]  {$\ISm{2}$};
\node (IS1m) [below of=IS2m] [statem] {$\ISm{1}$};
\node (IS0m) [below of=IS1m] [statem] {$\ISm{0}$};

\path[>=latex,*->] 
	  (epm3.east)	edge	[sloped]			node		{1}			(IS0m)
	  (epm2.east)	edge	[sloped,pos=.4]		node		{$1-x$}			(IS0m)
			edge	[sloped, near start]			node		{$x$}			(IS1m)
	  (epm1.east)	edge	[crossline] (IS0m)
			edge	[sloped, pos=.3]		node		{$1-x$}			(IS0m)
			edge	[sloped]		node		{$x$}			(IS2m)
			
;

\end{tikzpicture}
\caption[Transitions from $\ISm{2}$.]{Transitions from $\ISm{2}$. This figure states the expected food at each choice for the creature's next location ($C_{t+1}$) when in the information state $\ISm{2}$. It also depicts the transition probabilities to other information states based on $C_{t+1}$.}\label{fig:is2-transitions}
\end{center}
\end{figure}

\begin{figure}
\begin{center}
\begin{tikzpicture}
[inner sep=2mm,node distance=2.5cm,auto, line width=1pt, shorten <=-4pt,
statep/.style={circle, shade, top color=white,
    bottom color=yellow!50!black!20, draw=yellow!40!black!60, very
    thick,minimum size=8mm},
statem/.style={circle, shade, top color=white,
    bottom color=blue!50!black!20, draw=blue!40!black!60, very
    thick,minimum size=8mm},
 location/.style={rectangle, rounded corners, shade, top color=white,
    bottom color=green!50!black!20, draw=green!40!black!60, very
    thick,text width=4cm,text badly centered },
crossline/.style={preaction={draw=white, -,shorten <=10pt,
line width=6pt}}
]
\node (epmd) at (0,0) [anchor=east, location] {$C_{t+1}=e\pm d$};
\node (epmdpm1) at (0,-1.5) [anchor=east, location] {$C_{t+1}=e\pm d\pm1$};

\node (IS1m) at (4,.35) [statem] {$\ISm{1}$};
\node (IS0m) [below of=IS1m] [statem] {$\ISm{0}$};

\path[>=latex,*->] 
	  (epmd.east)	edge	[sloped, pos=.4]		node		{$1-x$}			(IS0m)
			edge	[sloped]			node		{$x$}			(IS1m)
	  (epmdpm1.east)	
			edge	[sloped]		node		{$1$}			(IS0m)
;
\end{tikzpicture}
\caption[Transitions from $\ISm{3}$.]{Transitions from $\ISm{3}$. The expected food in this state is a function of $d$; the creature strategies do not spend more than one turn in this state so the effect on the expected food over time is neglible. It does depict the transition probabilities to other information states based on $C_{t+1}$.}\label{fig:is3-transitions}
\end{center}
\end{figure}
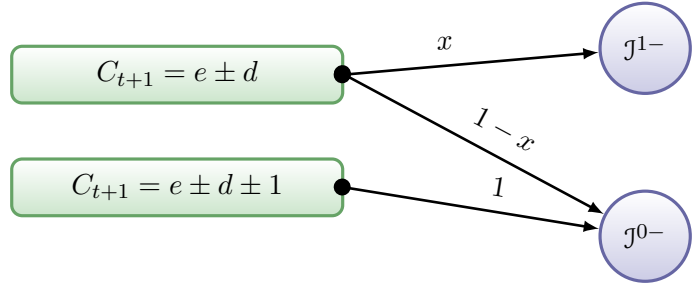

\subsection{Strategies Defined by Information States}
\label{sec:formal_strategies}

As defined earlier, a Markov algorithm is any strategy that (perhaps stochastically) maps all possible PMFs to a choice for the creature's next location.
The Infomax, Max-Likelihood, and Hybrid algorithms are depicted in Figure \ref{fig:strategies}. 
It is straightforward to prove (see Lemma \ref{lem:infomaxstrat}) that the Infomax strategy is in fact information optimal in the appropriate sense.
For comparison, we may consider the Max-Likelihood strategy, under which the creature always chooses to move to the location at which the injector is most likely to be situated after its next move (with a coin toss in case of a tie).  
The transitions depicted in Figure \ref{fig:strategies} define the Max-Likelihood strategy provided the parameter $x$ describing  the source random walk $P_x$ is between $1/4$ and $1/3$. 
The Hybrid strategy, also depicted in Figure \ref{fig:strategies}, makes the same choices as the Infomax strategy in the states $\ISm{1}$, $\ISm{2}$, and $\ISm{3}$; however, it makes the same choice as the Max-Likelihood strategy in state $\ISm{0}$.


\begin{figure}
\begin{multicols}{2}
\centering
$\ISm{0}$\\
\begin{tikzpicture}[xscale=.4,y=4cm]
\foreach \xpos / \xheight in
{-1/.3, 0/.4, 1/.3}
{
\draw[ycomb,
color=gray,
line width=0.15cm]
(0\xpos,0) -- (\xpos,\xheight);
}
\draw (-1.2,0) -- (1.2,0);
\draw[color=red,very thick] (-1,.15) ellipse (.4cm and .8cm);
\draw[color=blue,rotate=15,very thick] (.2,.2) ellipse (.45cm and 1.1cm);
\draw[color=darkgreen,rotate=-15,very thick] (-.2,.2) ellipse (.45cm and 1.1cm);
\node[anchor=north] at (0,-.05) {e};
\end{tikzpicture}

\vspace{.5cm}

\centering
$\ISm{2}$\\
\begin{tikzpicture}[xscale=.4,y=8cm]
\foreach \xpos / \xheight in
{-3/.15, -2/.2, -1/.15,3/.15, 2/.2, 1/.15 }
{
\draw[ycomb,
color=gray,
line width=0.15cm]
(0\xpos,0) -- (\xpos,\xheight);
}
\draw (-3.2,0) -- (3.2,0);

\draw[color=red,very thick] (-3,.075) ellipse (.4cm and 1cm);
\draw[color=darkgreen,very thick] (3,.075) ellipse (.4cm and 1cm);
\draw[color=blue,very thick] (-2,.1) ellipse (.4cm and 1.1cm);
\node[anchor=north] at (0,-.02) {e};
\draw (0,.01) -- (0,-.01);
\end{tikzpicture}

\columnbreak
\centering
$\ISm{1}$\\
\begin{tikzpicture}[xscale=.4,y=5.33cm]
\foreach \xpos / \xheight in
{-2/.15, -1/.2, 0/.3, 1/.2, 2/.15}
{
\draw[ycomb,
color=gray,
line width=0.15cm]
(0\xpos,0) -- (\xpos,\xheight);
}
\draw (-2.2,0) -- (2.2,0);
\draw[color=red,very thick] (-2,.075) ellipse (.4cm and .6cm);
\draw[color=darkgreen,very thick] (2,.075) ellipse (.4cm and .6cm);
\draw[color=blue,very thick] (0,.15) ellipse (.4cm and 1.1cm);
\node[anchor=north] at (0,-.05) {e};
\end{tikzpicture}

\vspace{.5cm}
\centering
$\ISm{3}$\\
\begin{tikzpicture}[xscale=.4,y=8cm]
\foreach \xpos / \xheight in
{-3.5/.15, -2.5/.2, -1.5/.15,3.5/.15, 2.5/.2, 1.5/.15 }
{
\draw[ycomb,
color=gray,
line width=0.15cm]
(0\xpos,0) -- (\xpos,\xheight);
}
\draw (-3.7,0) -- (3.7,0);
\draw[color=red,very thick] (-1.5,.075) ellipse (.4cm and 1cm);
\draw[color=darkgreen,very thick] (1.5,.075) ellipse (.4cm and 1cm);
\draw[color=blue,very thick] (-2.5,.1) ellipse (.4cm and 1.1cm);
\node[anchor=north] at (0,-.02) {e};
\draw (0,.01) -- (0,-.01);
\end{tikzpicture}
\end{multicols}
\caption[Strategy Definitions]{The {\color{red} Infomax (red)}, {\color{blue} Max-Likelihood (blue)}, and {\color{darkgreen}  Hybrid (green)} strategies are defined by the choice of location given $\pmfm{}$ is in a particular information state. 
In each panel above, a colored oval indicates the location to which a creature following the given strategy would move at time $t+1$, starting from the corresponding information state at time $t$.
For example,
if the creature is certain that the source is currently at a given location (say $e$), \textit{i.e.}~the creature is in information state $\ISp{0}$ (Top Left panel), then under  
the {\color{blue} Max-Likelihood} and {\color{darkgreen} Hybrid} strategies the creature would move to location $e$, the most likely location for the source to be at the next step (provided the probability of the source not moving is greater than $1/3$ or equivalently $x<1/3$).  But the {\color{red} Infomax} strategy would move to $e\pm 1$ with equal likelihood, because this location minimizes the expected uncertainty about the source's location at the next time step.   The remaining panels are interpreted similarly.
In the information state $\ISm{1}$ (Top Right panel), the injector is more likely to be at the location $e$ than at the locations $e\pm1$ if and only if $x>1/4$. Therefore the Max-Likelihood \emph{algorithm} as depicted in this figure is consistent with the likelihood maximization \emph{strategy} only for $1/4<x<1/3$. 
By symmetry, if a strategy would choose the location $e+m$ with $m\in\{0,\cdots,N-1\}$, the transition probabilities to different states and the expected food would be the same at location $e-m$. 
Therefore, the set of strategies that choose $e+m$ with probability $p\in[0,1]$ and $e-m$ with probability $1-p$ in a given information state make equivalent choices for  $C_{t+1}$ with respect to the information-state coarse graining.
}\label{fig:strategies}
\end{figure}
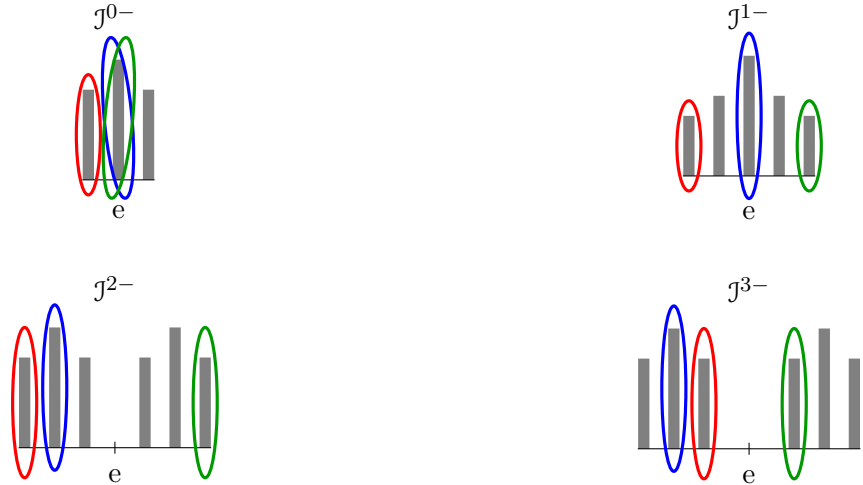

After two observations, all three strategies enter the absorbing states. In the long term, there is a non-zero probability that the creature is in each absorbing state, and the creature ``almost certainly'' remains in the absorbing set. 
These states are highlighted for each strategy in Figure \ref{fig:strategy_transitions}. As the transition probabilities for these finite, absorbing, coarse-grained states are exact, it is possible to calculate the expected time the creature will spend in each state (see Appendix \ref{app:ergodic}).
%


\begin{figure}
 
\begin{multicols}{2}
\subsection*{Strategy Definitions}

\begin{multicols}{2}
\centering
$\ISm{0}$\\
\begin{tikzpicture}[xscale=.4,y=4cm]
\foreach \xpos / \xheight in
{-1/.3, 0/.4, 1/.3}
{
\draw[ycomb,
color=gray,
line width=0.15cm]
(0\xpos,0) -- (\xpos,\xheight);
}
\draw (-1.2,0) -- (1.2,0);
\draw[color=red,very thick] (-1,.15) ellipse (.4cm and .8cm);
\draw[color=blue,rotate=15,very thick] (.2,.2) ellipse (.45cm and 1.1cm);
\draw[color=darkgreen,rotate=-15,very thick] (-.2,.2) ellipse (.45cm and 1.1cm);
\end{tikzpicture}

\vspace{.5cm}

\centering
$\ISm{2}$\\
\begin{tikzpicture}[xscale=.4,y=8cm]
\foreach \xpos / \xheight in
{-3/.15, -2/.2, -1/.15,3/.15, 2/.2, 1/.15 }
{
\draw[ycomb,
color=gray,
line width=0.15cm]
(0\xpos,0) -- (\xpos,\xheight);
}
\draw (-3.2,0) -- (3.2,0);

\draw[color=red,very thick] (-3,.075) ellipse (.4cm and 1cm);
\draw[color=darkgreen,very thick] (3,.075) ellipse (.4cm and 1cm);
\draw[color=blue,very thick] (-2,.1) ellipse (.4cm and 1.1cm);
\end{tikzpicture}

\columnbreak
\centering
$\ISm{1}$\\
\begin{tikzpicture}[xscale=.4,y=5.33cm]
\foreach \xpos / \xheight in
{-2/.15, -1/.2, 0/.3, 1/.2, 2/.15}
{
\draw[ycomb,
color=gray,
line width=0.15cm]
(0\xpos,0) -- (\xpos,\xheight);
}
\draw (-2.2,0) -- (2.2,0);
\draw[color=red,very thick] (-2,.075) ellipse (.4cm and .6cm);
\draw[color=darkgreen,very thick] (2,.075) ellipse (.4cm and .6cm);
\draw[color=blue,very thick] (0,.15) ellipse (.4cm and 1.1cm);
\end{tikzpicture}

\vspace{.5cm}
\centering
$\ISm{3}$\\
\begin{tikzpicture}[xscale=.4,y=8cm]
\foreach \xpos / \xheight in
{-3.5/.15, -2.5/.2, -1.5/.15,3.5/.15, 2.5/.2, 1.5/.15 }
{
\draw[ycomb,
color=gray,
line width=0.15cm]
(0\xpos,0) -- (\xpos,\xheight);
}
\draw (-3.7,0) -- (3.7,0);
\draw[color=red,very thick] (-1.5,.075) ellipse (.4cm and 1cm);
\draw[color=darkgreen,very thick] (1.5,.075) ellipse (.4cm and 1cm);
\draw[color=blue,very thick] (-2.5,.1) ellipse (.4cm and 1.1cm);
\end{tikzpicture}
\end{multicols}

\subsection*{{\color{darkgreen}Hybrid Creature}}
\centering
\begin{tikzpicture}
[inner sep=.5mm,node distance=2.5cm,auto, line width=1pt, >=stealth,
statem/.style={circle, shade, top color=white,
    bottom color=blue!50!black!20, draw=blue!40!black!60, very
    thick,minimum size=5mm},
statema/.style={circle split, shade, top color=white,
    bottom color=blue!50!black!20, draw=darkgreen, very
    thick,minimum size=5mm}
]

\node[statem] 	(iss)	at (0,0)			{$\ISm{*}$};
\node		(start) at (-1,0)			{Start};
\node[statem] 		(is3) [below of=iss] 		{$\ISm{3}$};
\node[statem] 		(isi) at (0,1.75) 		{$\ISm{I}$};
\node[statema] 		(is1) [right of=iss] 		{$\ISm{1}$\nodepart{lower} {$\frac{2x}{1+2x}$}};
\node[statem] 		(is2) [right of=isi] 		{$\ISm{2}$};
\node[statema] 		(is0) [below of=is1] 		{$\ISm{0}$\nodepart{lower} {$\frac{1}{1+2x}$}};

\path[->]
(start) edge[shorten >=1pt] (iss)
(iss) edge node[pos=.3] {$\frac{4}{N}$} (is2)
(iss) edge node {$\frac{4}{N}$} (is1)
(iss) edge node[pos=.3] {$\frac{2}{N}$} (is0)
(iss) edge node {$\frac{N-10}{N}$} (is3)
(is3) edge node{$1$} (is0); 

\path[->]
(is0) edge [loop below] node {$1-2x$} ()
(is1) edge [left,bend left] node {$1$} (is0)
(is2.east) edge [bend left] node {$1$} (is0.east);

\path[->]
(is0) edge [right,bend left] node {$2x$} (is1);

\end{tikzpicture}

\subsection*{{\color{red}Information Theory Creature}}   
\centering
\begin{tikzpicture}
[inner sep=.5mm,node distance=2.5cm,auto, line width=1pt, >=stealth,
statem/.style={circle, shade, top color=white,
    bottom color=blue!50!black!20, draw=blue!40!black!60, very
    thick,minimum size=5mm},
statema/.style={circle split, shade, top color=white,
    bottom color=blue!50!black!20, draw=red, very
    thick,minimum size=5mm}
]

\node[statem] 	(iss)	at (0,0)			{$\ISm{*}$};
\node		(start) at (-1,0)			{Start};
\node[statem] 		(is3) [below of=iss] 		{$\ISm{3}$};
\node[statem] 		(isi) at (0,1.75) 		{$\ISm{I}$};
\node[statem] 		(is1) [right of=iss] 		{$\ISm{1}$};
\node[statem] 		(is2) [right of=isi] 		{$\ISm{2}$};
\node[statema] 		(is0) [below of=is1] 		{$\ISm{0}$ \nodepart{lower} {$1$}};

\path[->]
(start) edge[shorten >=1pt] (iss)
(iss) edge node[pos=.3] {$\frac{4}{N}$} (is2)
(iss) edge node {$\frac{4}{N}$} (is1)
(iss) edge node[pos=.3] {$\frac{2}{N}$} (is0)
(iss) edge node {$\frac{N-10}{N}$} (is3)
(is3) edge node{$1$} (is0);

\path[->]
(is0) edge [loop below] node {1} ()
(is1) edge [bend left] node {1} (is0)
(is2.east) edge [bend left] node {1} (is0.east);

\end{tikzpicture}

\subsection*{{\color{blue}Maximum Likelihood Creature}}
\centering
\begin{tikzpicture}
[inner sep=.5mm,node distance=2.5cm,auto, line width=1pt, >=stealth,
statem/.style={circle, shade, top color=white,
    bottom color=blue!50!black!20, draw=blue!40!black!60, very
    thick,minimum size=5mm},
statema/.style={circle split, shade, top color=white,
    bottom color=blue!50!black!20, draw=blue, very
    thick,minimum size=5mm},
crossline/.style={preaction={draw=white, -,
line width=6pt}}
]

\node[statem] 	(iss)	at (0,0)			{$\ISm{*}$};
\node		(start) at (-1,0)			{Start};
\node[statem] 		(is3) [below of=iss] 		{$\ISm{3}$};
\node[statem] 		(isi) at (0,1.75) 		{$\ISm{I}$};
\node[statema] 		(is1) [right of=iss] 		{$\ISm{1}$\nodepart{lower} {$\frac{2}{4 + x}$ }};
\node[statema] 		(is2) [right of=isi] 		{$\ISm{2}$\nodepart{lower} {$\frac{2x}{4 + x}$ }};
\node[statema] 		(is0) [below of=is1] 		{$\ISm{0}$\nodepart{lower} {$\frac{2-x}{4 + x}$ }};

\path[->]
(start) edge[shorten >=1pt] (iss)
(iss) edge node[pos=.3] {$\frac{4}{N}$} (is2)
(iss) edge node {$\frac{4}{N}$} (is1)
(iss) edge node[pos=.3] {$\frac{2}{N}$} (is0)
(iss) edge node[left] {$\frac{N-10}{N}$} (is3)
(is3) edge[crossline] (is1)
(is3) edge node[pos=.3]{$x$} (is1)
(is3) edge node{$1-x$} (is0);

\path[->]
(is0) edge [loop below] node {$1-2x$} ()
(is1) edge [left,bend left] node {$x$} (is0)
(is2.east) edge [bend left=60] node[pos=.85] {$1-x$} (is0.east);

\path[->]
(is1) edge [loop right] node [pos=.15, above] {$1-2x$} ()
(is0) edge [right, bend left] node {$2x$} (is1)
(is2) edge [bend left] node {$x$} (is1)
(is2) edge [bend left] node {$x$} (is1);

\path[->]
(is1) edge [bend left] node {$x$} (is2);

\end{tikzpicture}

\end{multicols}

\caption[Strategy Transitions]{Transition Graphs for each of the strategies.  Top left panel recapitulates Figure \ref{fig:strategies}.  Other panels depict transitions between information states for each strategy as labeled.  Each creature begins in state $\ISm{*}$.  Arrow labels indicate transition probabilities.  Large colored circles  denote positive recurrent states for each strategy, with expressions provided for the equilibrium probability of occupying the state. Note the sum of transition probabilities  on each set of arrows pointing away from a given node sum to one.  For example, the creature always starts in state $\ISm{*}$, the uniform distribution, representing complete ignorance of the injector's location.  From this state it moves to state $\ISm{2}$ with probability $4/N$, to state $\ISm{1}$ with probability $4/N$, to state $\ISm{0}$ with probability $2/N$, and to state $\ISm{3}$ with probability $(N-10)/N$.  These transitions depend only on making the first observation, and are the same for each strategy.  Subsequent transition probabilities differ according to the strategy employed, which determines how the creature moves relative to the source.}\label{fig:strategy_transitions}
\end{figure}
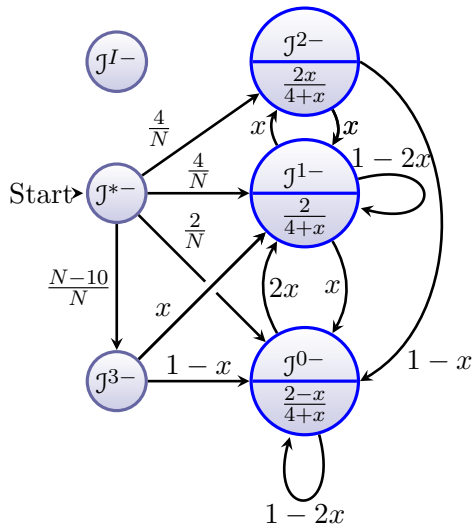

\subsection{Strategies Compared}
\label{sec:asymptotic}


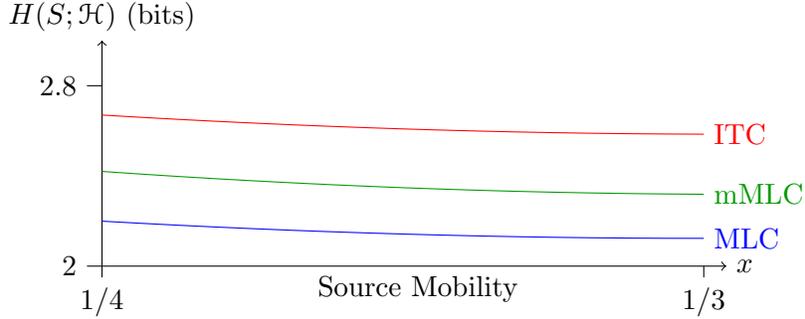
\begin{figure}
\label{fig:strategy_information}
\centering
\begin{tikzpicture}[domain=.25:.33, xscale=100, yscale=3]
    \draw[->] (.25,2) -- (.292,2) node[below] {Source Mobility} --(.333,2) node[right] {$x$};
    \draw[->] (.25,2) -- (.25,3) node[above] {$H(S;\Hist{})$ (bits)};
    \draw (.25,2) -- (.25,1.95) node [below] {1/4};
    \draw (.33,2) -- (.33,1.95) node [below] {1/3};
    \draw (.25,2) -- (.248,2) node [left] {$2$};
    \draw (.25,2.8) -- (.248,2.8) node [left] {$2.8$};
    \draw[color=red] plot file {IITC} 
        node[right] {ITC};
    \draw[color=darkgreen] plot file {ImMLC} 
        node[right] {mMLC};
    \draw[color=blue] plot file {IMLC} 
        node[right] {MLC};
\end{tikzpicture} 
\caption[Expected Information]{The expected information for each strategy as a function of $x$, the source mobility parameter.   Larger values of  $x$ correspond to greater unpredictability of future source locations. Each curve represents the expected decrease in uncertainty about the source position when conditioned on the creature's history and observations, compared with the underlying uncertainty in the absence of any observations  ($\log_2(N)$ bits), for the different strategies considered.
As expected, the movement strategy based on information maximization (ITC, Information Theory Creature) dominates the maximum likelihood based strategy (MLC, Maximum Likelihood Creature) and the hybrid strategy (mMLC, modified Maximum Likelihood Creature) over the entire range of $x$ considered.  
Also as expected, the information about the source position relative to the creature's position that is gained through observation decreases as the source's inherent unpredictability increases, \textit{i.e.}~all three curves have strictly negative slope.}
\end{figure}
The information a creature has about the injector's location given an information state is calculated in Lemma \ref{lem:infomaxstrat}. Using the expected time spent in each state for a given strategy\footnote{The fraction of time a creature spends in a given information state is exactly the equilibrium probability distribution, derived in Appendix \ref{app:strategytransitionmatrices}. See also Figure  \ref{fig:strategy_transitions}, Table 2.2 
 and Figure \ref{fig:fraction-of-time}).} the expected information for each strategy is shown in Figure 2.9. 
This figure shows the strict ranking of information theoretic performance between the strategies; namely, maximum likelihood has the least, the hybrid strategy is in the middle, and Infomax is the best. 
This result, which is not unexpected, is easily explained by examining the fraction of time a creature spends in a given information state.  The uncertainty about the source's location as quantified by the entropy of the source location distribution (see Table \ref{tab:entropies}) is strictly increasing as a function of both $x$ and the information state index $k$.  The creature in information state $\ISp{0}$ knows exactly the current location of the food source; creatures in $\ISp{1}$ and $\ISp{2}$ have progressively more uncertainty. The information-seeking creature is able to remain in state $\ISp{0}$ on every time step.  
In Lemma \ref{lem:infomaxstrat} we prove that the Infomax strategy is information optimal, in the sense that the creature following the Infomax strategy always has the least uncertainty, \textit{i.e.}~the most information, about  the source's location.
For all three strategies as the source mobility  $x$ increases, the creature's uncertainty about the source location gradually grows, as seen in the decrease of the mutual information in Figure 2.9. 
Table 2.2 
shows the fraction of time spent by each strategy in each information state; Figure \ref{fig:fraction-of-time} displays this comparison graphically.

\begin{table}\label{tab:fraction-of-time}
\centering
\renewcommand{\arraystretch}{1.5}
\begin{tabular}{c|c|ccc|ccc}
$k$ & ITC & $x\to 1/4$ & Hybrid & $x\to 1/3$ &$x\to 1/4$ & MLC & $x\to 1/3$\\ 
0 & 1 & $\frac{2}{3}\le$ & $\frac{1}{1+2x}$ & $\le \frac{3}{5} $ & $\frac{7}{17}\approx 0.41 \le $& $ \frac{2-x}{4+x}$ & $ \le \frac{5}{13}\approx 0.38 $
\\
1 & 0 &  $\frac{1}{3}\le$ & $\frac{2x}{1+2x}$ & $\le\frac{2}{5}$ & $\frac{8}{17}\approx 0.47\le$ & $\frac{2}{4+x}$ & $\le\frac{6}{13}\approx 0.46$
\\
2 & 0 & & 0 & & $\frac{2}{17}\approx 0.12\le$ & $\frac{2x}{4+x}$ & $\le\frac{2}{13}\approx 0.15$
\\
3+& 0 & & 0 & & & 0 & 
\end{tabular}
\renewcommand{\arraystretch}{1}
\caption[Fraction of time spent in each information state.]{Fraction of time spent in each information state for each strategy.  Each row indicates the fraction of time $\pi_k$ the  creature spends in the corresponding information state as a function of $x$, and the lower and upper limits of $\pi_k$ viewed as a function of $x$. The Information Theory Creature (ITC) remains permanently in the maximally informed state, $k=0$.  The Hybrid creature (mMLC) spends twice as much time in state $0$ when $x$ is small as it does in the less informed state $1$.  As $x$ increases from $x\approx 1/4$ to $x\approx 1/3$, the creature spends only 50\% as much time in state $0$ as in state $1$.   The pure MLC creature always spends \emph{more} time in state $0$ than in state $1$, and an increasing amount of time in state 2 as $x$ increases. 
 Compare Figure \ref{fig:fraction-of-time}.
}
\end{table}

\begin{figure}[htbp] 
   \centering
   \includegraphics[width=6in]{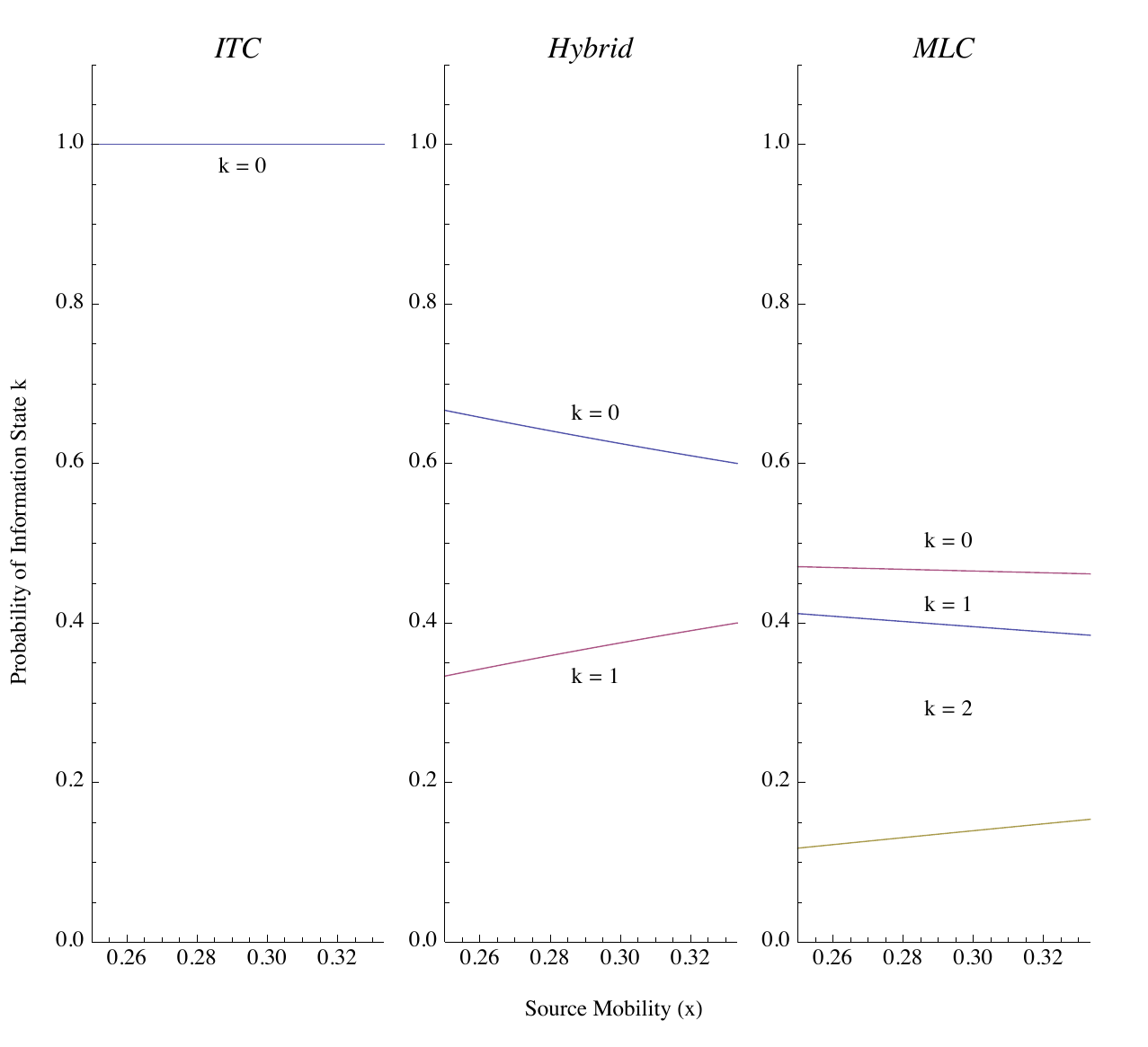} 
 \caption[Fraction of time spent in each information state.]{Fraction of time spent in the $k^{th}$ information state for each strategy, as a function of the source mobility $x$.  In state $\ISp{k}$ the food source is known to be at position $e\pm k$ for some location $e$. \textit{ITC}, Information Theory Creature.  The creature following this algorithm remains in information state zero for all time after an initial transient. \textit{Hybrid}, modified Maximum Likelihood Creature.  The creature following this algorithm asymptotically spends all of its time in information states zero and one.  The fraction of time spent in the lowest entropy state (state zero) decreases as the source mobility increases.  \textit{MLC}, Maximum Likelihood Creature.     The creature following this algorithm asymptotically spends its time in information states zero, one and two. Compare Table 2.2.}   
\label{fig:fraction-of-time}
\end{figure}

Figure \ref{fig:strategy_food} shows the expected food gain for each strategy.
Unlike the expected information, \emph{the expected food shows no consistent ranking between strategies.} For some choice of source mobility in the range $1/4<x<1/3$ each strategy outperforms the other two.   In particular, the infomax strategy provides a better average food intake than either the maximum likelihood or the hybrid strategy, when the source is the most unpredictable ($x\lesssim 1/3$), but fails to provide the best performance as the next source location becomes easier to predict, given the current location.


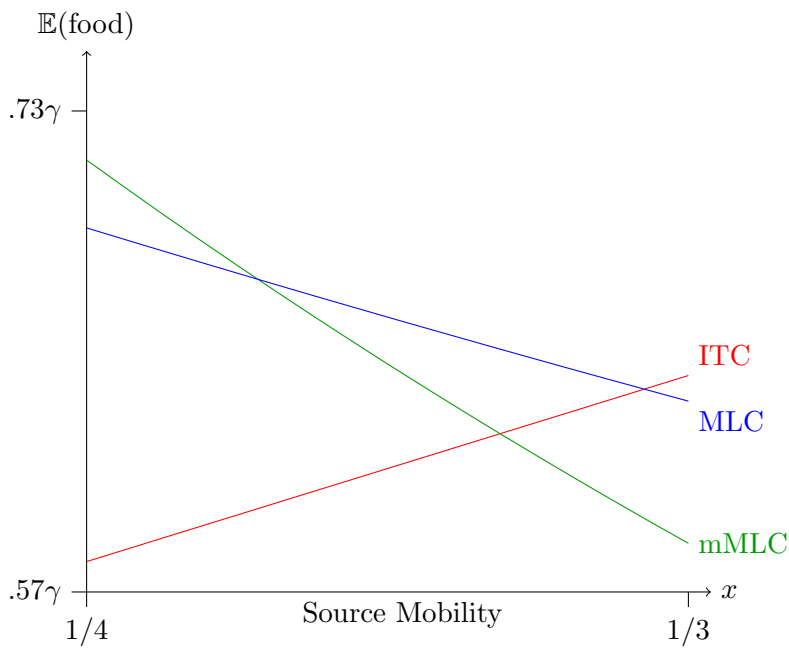
\begin{figure}
 \centering
     \begin{tikzpicture}[domain=.25:.33, xscale=100, yscale=40]
    \draw[->] (.25,.57) -- (.292,.57) node[below] {Source Mobility} -- (.333,.57) node[right] {$x$};
    \draw[->] (.25,.57) -- (.25,.75) node[above] {$\E(\text{food})$};
    \draw (.25,.57) -- (.25,.565) node [below] {1/4};
    \draw (.33,.57) -- (.33,.565) node [below] {1/3};
    \draw (.25,.57) -- (.248,.57) node [left] {$.57\gamma$};
    \draw (.25,.73) -- (.248,.73) node [left] {$.73\gamma$};
    \draw[color=red] plot file {EFITC} 
        node[above right] {ITC};
    \draw[color=darkgreen] plot file {EFmMLC} 
        node[right] {mMLC};
    \draw[color=blue] plot file {EFMLC} 
        node[below right] {MLC};
\end{tikzpicture}           
\caption[Expected Food]{The expected food for each strategy as a function of the source mobility parameter $x$. (ITC, Information Theory Creature.  MLC, Maximum Likelihood Creature. mMLC, modified Maximum Likelihood Creature, or Hybrid.)  Expected food is calculated as described in Section \ref{sec:markov}.  Each strategy is superior to the other two (in terms of long-term food gain) for some range of source mobility.  
Over the range considered ($1/4<x<1/3$), the infomax strategy dominates only in a narrow range corresponding to relatively random motion of the source.  
The maximum likelihood based strategy dominates over a broad range of mobilities, but for the lowest mobilities the hybrid strategy is the most successful.
Compare Figure 2.9.}
\label{fig:strategy_food}
\end{figure}

It is tempting to explain this result by noting that when $x\lesssim 1/3$, the food source itself generates an ensemble of trajectories that has a higher entropy generation rate then when $x$ is smaller.  Perhaps these more ``information rich" trajectories somehow cause a creature focused on \emph{information about} the food source to outperform those focused only on \emph{colocation with} the source.  Unfortunately such qualitative reasoning would fail to explain why, for intermediate values of the source mobility, it is the maximum likelihood creature that gains more food, not the hybrid  that mixes the ML and IT strategies.   
Also puzzling is the trend of the ITC curve. It is not surprising that the performance of the MLC and mMLC strategies decreases as the unpredictability of the source increases, but why should the ITC performance curve have a \emph{positive} slope?

To understand more fully how ``information'' and ``utility'' intertwine to determine the performance of each strategy, we may consider the transition probabilities for the Markovian random walk on the set of states \emph{jointly} describing the creature's information state and its actual distance from the food source.  In a sense this takes us a step backwards as we ``un-coarse-grain" the system.  With the benefit of hindsight, we take advantage of the fact that for each strategy, the  underlying system admits a Markovian coarse graining in terms of the information states alone.
Equations (\ref{eq:joint-trans-ITC}-\ref{eq:joint-trans-MLC}) in Appendix \ref{app:joint} provides the transition probabilities for joint (distance, information) states for each strategy.  
Examination of these transition probabilities shows that two of the strategies (ITC and MLC) admit an exact coarse graining in terms of the distance to the source.  
Table 2.3 
provides the resulting distribution of time spent at each distance from the source, for each strategy; these results are plotted in Figure \ref{fig:distance-probs}.

\begin{table}\label{tab:distance-probs}
\centering
\begin{tabular}{r|rrr}
Distance & ITC & Hybrid & MLC \\ \hline
0 & $x$ 	& $(1-x)^2/(1+2x)$ 	& $(2-2x)/(4+x)$ \\
1 & $1-2x$ 	& $(3x-2x^2)/(1+2x)$ 	& $2/(4+x)$ \\
2 & $x$ 	& $2x^2/(1+2x)$ 	&  $2x/(4+x)$ \\
3 & 0 		& $(x-2x^2)/(1+2x)$ 	&  $x^2/(4+x)$ \\
4 & 0 		& $x^2/(1+2x)$ 		&  $(x-2x^2)/(4+x)$ \\
5 & 0 		& 0 			&  $x^2/(4+x)$ \\
$\ge 6$ &  0 	& 0 			& 0
\end{tabular}
\caption[Fraction of time spent a given distance from the source.]{Fraction of time spent by each creature a given absolute distance from the food source, as a function of the mobility parameter $x$. ITC, Information Theory Creature; Hybrid, modified Maximum Likelihood Creature; MLC; Maximum Likelihood Creature.  Compare Figure \ref{fig:distance-probs}. }
\end{table}

\begin{figure}[htbp] 
   \centering
   \includegraphics[width=6in]{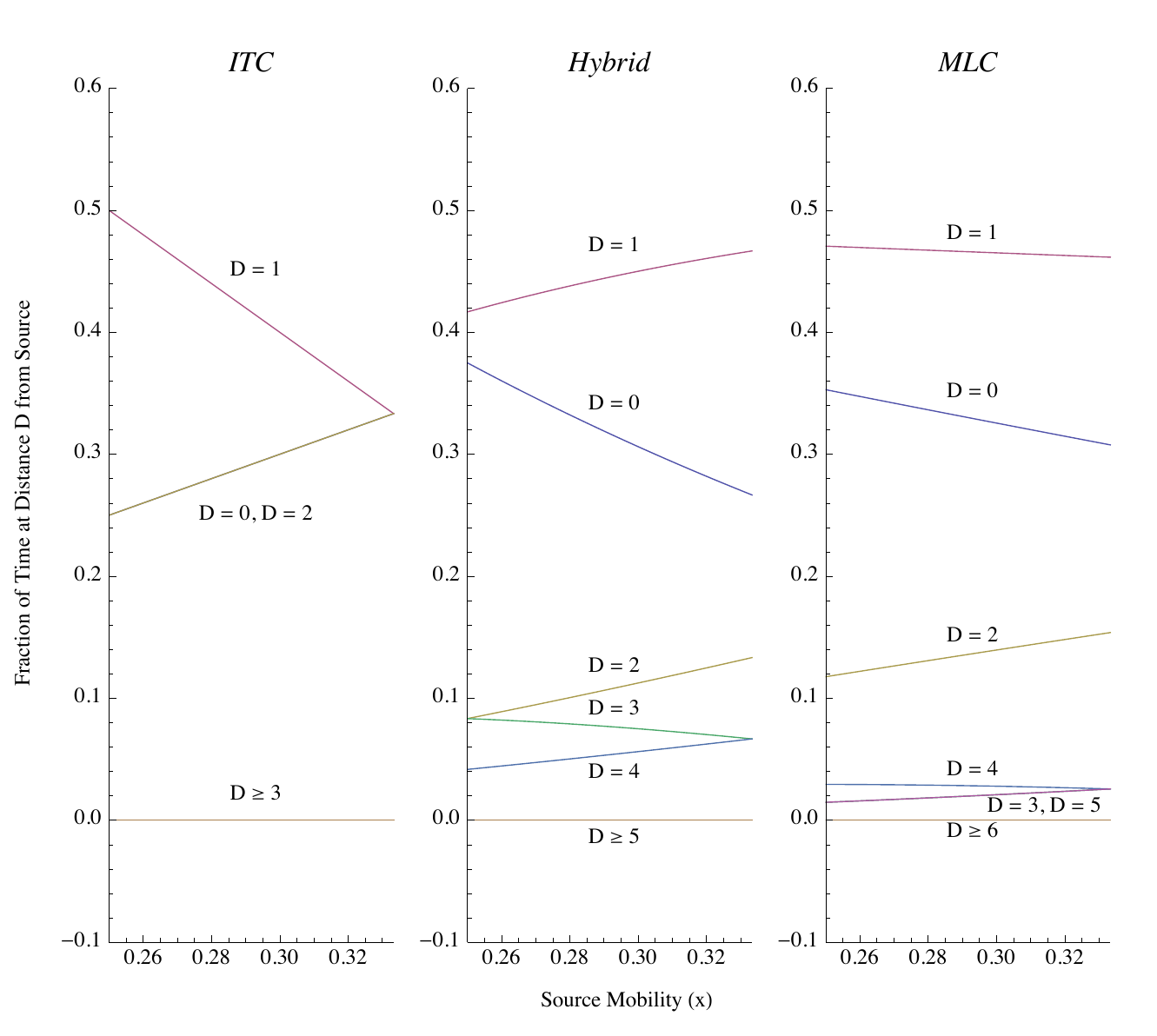} 
   \caption{Fraction of time spent by the creature a given absolute distance from the food source for each strategy, as a function of the mobility parameter $x$.  \textit{ITC}, Information Theory Creature.  Under this algorithm, the creature spends equal amounts of time at distance $D=0$ and $D=2$ from the source. The amount of time at the highest food concentration ($D=0$) increases as  the mobility parameter $x$ increases, while the amount of time spent one step away from the source ($D=1$) decreases. The probability of the ITC being a distance $D\ge 3$ from the source is zero, after the initial transient. \textit{Hybrid}, modified Maximum Likelihood Creature.  As $x$ increases, the mMLC spends more time at distance $D=1$ and less time at distance $D=0$.  The probability of the mMLC being a distance $D\ge 5$ from the source is zero, after the initial transient.   \textit{MLC}, Maximum Likelihood Creature.  As $x$ increases the MLC spends slightly less time at distances $D=1$ and $D=0$, and more time at distance $D=2$.   The probability of the MLC being a distance $D\ge 6$ from the source is zero, after the initial transient. Compare Table  2.3.}
   \label{fig:distance-probs}
\end{figure}

The expected food obtained under a given strategy is assumed to be the weighted average of the food available at each distance, weighted by the probability of being that distance from the source.  The food distribution for model II, as derived in Appendix \ref{app:food} (see also Figure \ref{fig:food-distribution}) is strictly decreasing in $d$.  Therefore the best conceivable performance would be for the creature to remain at zero distance from the source at all times.  The unpredictability of the source motion makes this impossible, and consequently the strategies vary in the fraction of time spent at distance $D=0$, as  a function of the source mobility parameter $x$.  As $x$ increases, both the Max-Likelihood and the Hybrid creatures spend less and less time colocalized with the source, while probability of $D=0$ increases for the Infomax strategy.  

The increase of utility for the ITC with increasing $x$ therefore has a simple, intuitive explanation.  The ITC acts in a way that takes into account the actual amount of food to be had at each destination only in a secondary fashion.  Aiming directly for the source's last location will lead, with probability $2x$, to a one bit increase in uncertainty about the source's location on the subsequent time step.  Therefore the IT creature will always choose to miss the most likely location by a single step.  It takes the second best position, from the point of view of food collection, in order to have the best position for continued certainty about the source's trajectory.    This strategy is a poor choice on its face, because the only way the creature will collect better than the second highest food amount is if it gets ``lucky" and the next move of the source happens to take it to the creature's location by chance.
The probability of colocation if the creature moves to a location one step to the right or left of the source's last known position is exactly the source mobility $x$.  Hence as $x$ increases the probability of collocation for the IT creature grows rather than decreases.  

Quantitatively, the expected food obtained by following the infomax strategy is straightforward to obtain.  If the food obtained a distance $d$ from the source is $\lambda_d$, then the expected food for the ITC may be written in terms of the second difference (discrete second derivative) of the food distribution:
\begin{eqnarray}
\E_{\mbox{ITC}}(\mbox{food})&=&x\lambda_0+(1-2x)\lambda_1+x\lambda_2\\
&=&\lambda_1+x(\lambda_0-2\lambda_1+\lambda_2) 
\end{eqnarray}
Therefore the increase in the expected food gained under the infomax strategy directly reflects the upwards concavity of the food distribution curve as shown in Figure \ref{fig:food-distribution}.  

To determine the effect of changing the shape of the food distribution curve about the source, we considered several generic food distribution functions $F(d)$, given in Table 2.4. 
\begin{table}\label{tab:food-distributions}
\begin{center}
\begin{tabular}{|cc|}\hline 
Type & $F(d), 0\le d \le 5$ \\ \hline
Concave Up & $(5-d)^2$\\
Concave Down & $25-d^2$\\
Linear & $25-5 d$\\
Delta Function &$25\, \delta(d)$\\ \hline
\end{tabular}
\caption[Generic dependence of food on distance.]{Examples of different shapes of food distribution functions: concave up, concave down, linear and delta function.  Each function has a nominal peak value of $F(0)=25$; each is taken to be zero for $d\ge 5$.  Dependence of expected food gained for each strategy for each function shown is plotted in Figure \ref{fig:food-distributions}.}
\end{center}
\end{table}
Three distributions are chosen representing ``concave up'', ``concave down'', or linear.  Each has constant second difference over the relevant distance range, respectively positive, negative, or zero.  Each is monotonically decreasing as a function of distance, with a peak at $d=0$.  Because the recurrent states only occupy distances in the range $0\le d \le 5$, we set our trial food distribution functions to zero for $d\ge 5$.  We also consider an all or none dependence of food gained on distance to the source (``delta function'' in Table 2.4). 
In the latter case food \emph{signaling} is decoupled from food \emph{obtained}.  Here it is assumed that the creature can determine the distance to the source by detecting the local concentration of a signal (visually or \textit{via} olfaction, for instance), but only captures the food through colocation.  This situation would be familiar to any predator. 

Figure \ref{fig:food-distributions} illustrates the results.   As anticipated, the slope of the expected food gain under the infomax strategy has the same sign as the second difference of the food distribution.
More surprisingly, the topology of the relative rankings of strategies as a function of source mobility appears to be relatively insensitive to the concavity of the food distribution. For the four cases  examined (Table 2.4) 
 the rankings appeared in the same order, with the hybrid strategy dominating the others for the lowest source mobilities, the maximum likelihood strategy dominating for intermediate values and the ITC dominating for the highest source mobilities.\footnote{The topology of relative rankings is not \emph{universally} invariant with respect to the food distribution, however. For instance, choosing $\lambda_0=5,\lambda_1=4,\lambda_2=1,\lambda_{k\ge 3}=0$ results in $\E_{ITC}(\mbox{food})<\E_{MLC}(\mbox{food})$ when $x=1/3$.}   The combined plots of strategy performance show the same ``triangle'' present for the distribution considered originally, compare Figure \ref{fig:food-distributions} and Figure \ref{fig:strategy_food}.\footnote{In order to see the full topology of the curves, we plot the expected food over a range of mobility values extending to lower mobility than the $x=1/4$ cutoff used in the rest of the paper.  For a discussion of the applicability of the different movement algorithms outside the range $1/4<x<1/3$, see Section \ref{ssec:strat-vs-alg}.}


\begin{figure}[htbp] 
   \centering
   \includegraphics[width=6in]{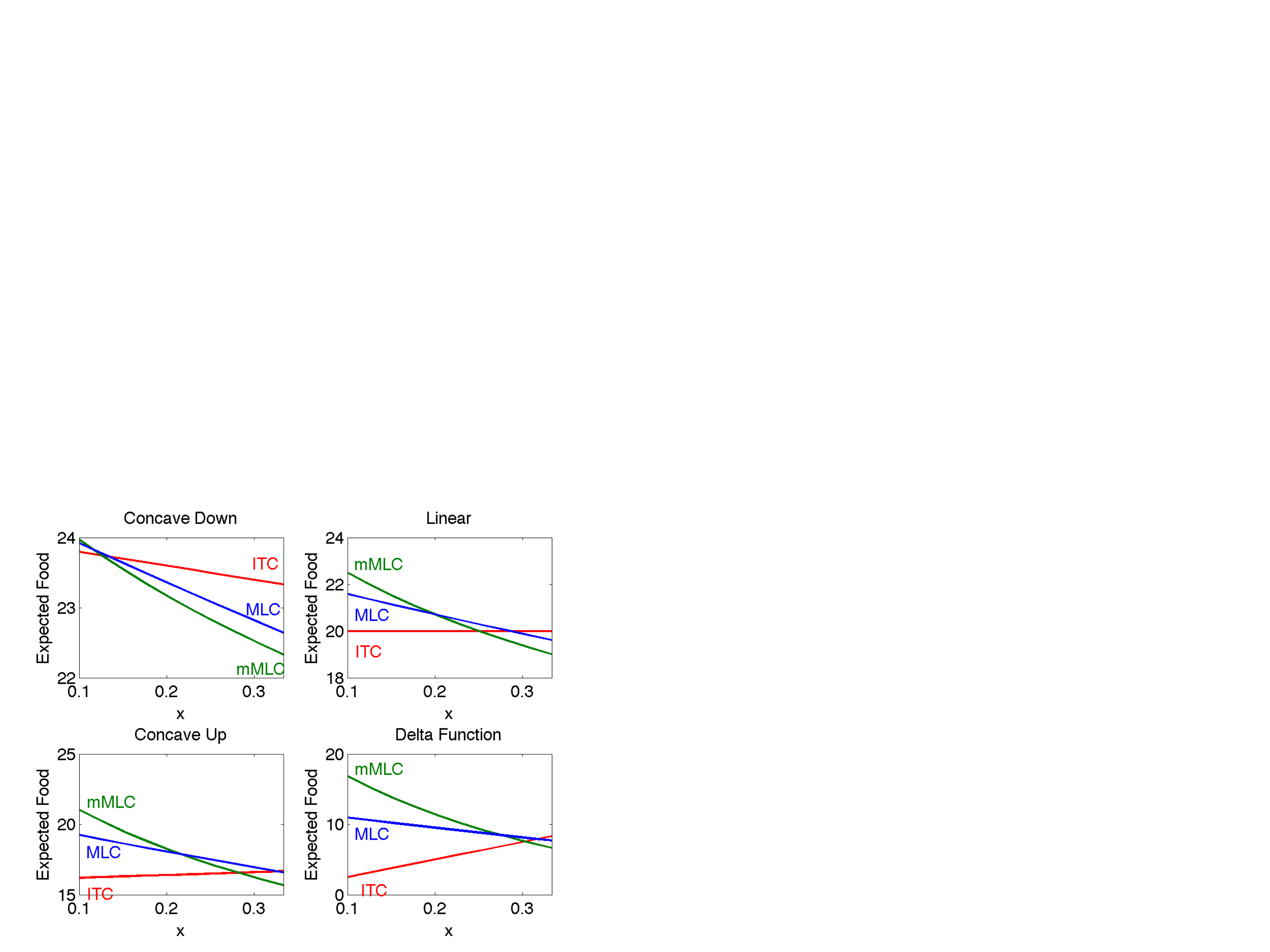} 
   \caption[Persistence of the ranking across food distributions]{Persistence of the ranking vis-a-vis expected food gathered across different food distribution shapes. Each plot shows expected food gained, in nominal units, versus source mobility on an expanded scale $1/10<x<1/3$.  Food gained as a function of distance from the source, $F(d)$, was taken to have constant second difference in three generic cases.  The nominal food function $F(d)$ was set to zero for $d\ge 5$ for this comparison.  Concave Down:  $F(d)=25-d^2$. Linear: $F(d)=25-5d$. Concave Up: $F(d)=(5-d)^2$.  Delta Function: $F(d)=25$ if $d=0$, otherwise $F=0$.  In each case, there is an upper range in which the Infomax movement algorithm (red curve) dominates, a middle range in which the Max-Likelihood movement algorithm (blue curve) dominates, and a lower range in which the Hybrid movement algorithm (green curve) dominates. In each case the curve for the Infomax algorithm is a straight line with slope either positive (for ``Concave Up'' and ``Delta Function'' food distributions), negative (for ``Concave Down'' food distribution) or zero (for ``Linear'' food distribution).}
   \label{fig:food-distributions}
\end{figure}

\subsection{Distance Traveled}

The distance traveled in pursuit of a strategy is another factor that might impact an organism's survival.  While detailed consideration of metabolic constraints would go beyond the scope of this paper, the models considered do lead to simple analytic expressions for two quantities of potential relevance in this context. In Appendix \ref{app:distance-traveled} we indicate how one calculates the expected distance moved per time step by a creature following each of the strategies, as well as the probability of movement per time step (as opposed to the probability of remaining motionless).  The average distance moved could be related directly or indirectly to a metabolic movement cost.  The probability of movement could also be related to a cost for initiating movement, such as the effort a foraging bird might have to exert in order to ascend to a given altitude, after which soaring arbitrary distances might incur relatively little additional metabolic cost, or to the differential risk of predation during a motile \textit{versus} a sessile phase of activity.  
\begin{table}\label{tab:distance-traveled}
\begin{center}
\renewcommand{\arraystretch}{2}
\begin{tabular}{lcccccr}\hline
&MLC&&ITC&&mMLC\\ \hline
$\E\left[\left|C_{t+1}-C_t\right|\right]$ & $\frac{8x}{4+x}$ & $<$ & $2x$ & $<$ & $\frac{6x}{1+2x}$ & (for $0<x\le1/2$)\\
$\Pr\left[C_{t+1}\ne C_t\right]$ & $\frac{4x}{4+x}$ & $<$ & $2x$& $<$ & $\frac{x(4-x)}{1+2x}$ & (for $0<x<0.4$)\\ \hline
\end{tabular}
\renewcommand{\arraystretch}{1}
\end{center}
\caption{Mean distance traveled and probability of moving per time step.  MLC: Maximum Likelihood Creature, ITC: Information Theory Creature, mMLC: modified Maximum Likelihood Creature.  
The mean, or expected distance traveled in one time step is $\E\left[\left|C_{t+1}-C_t\right|\right]$.  The probability of moving on a given time step is the probability that $C_t$ and $C_{t+1}$ take different values, \textit{i.e.}~$\Pr\left[C_{t+1}\ne C_t\right]$.  The last column shows the range of validity of the inequalities.  For the second row, the inequality MLC $<$ ITC holds for all $x>0$, but the inequality ITC $<$ mMLC reverses when $x>2/5$.   Compare Figure \ref{fig:distance-traveled}.}
\end{table}
\begin{figure}[htbp] 
   \centering
\includegraphics[width=6in]{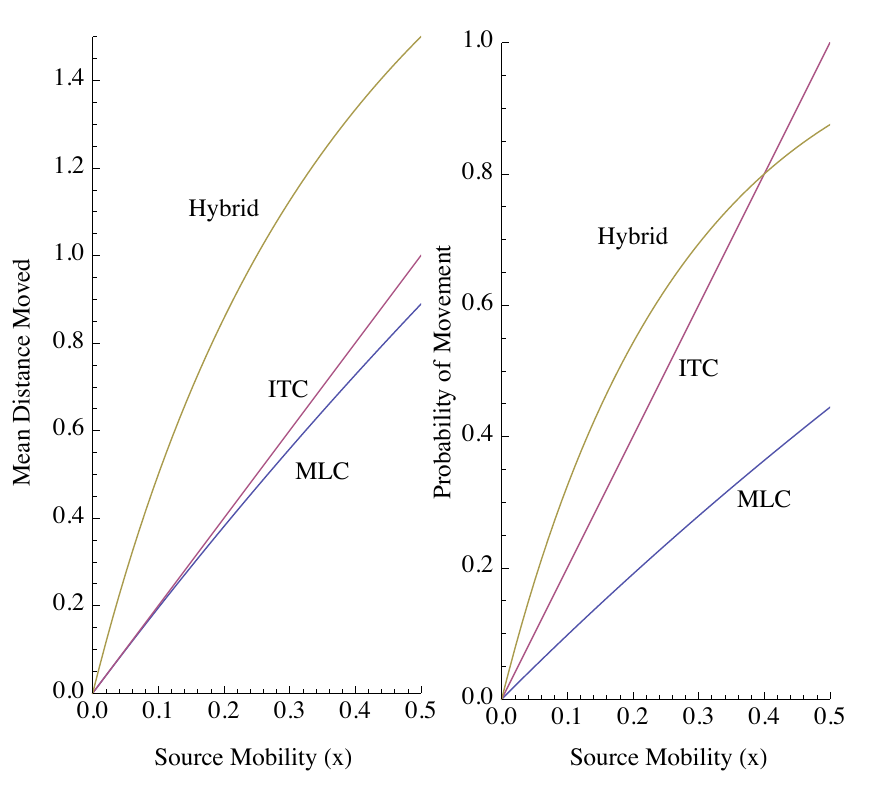}
   \caption{Mean distance traveled and probability of moving per time step. MLC: Maximum Likelihood Creature, ITC: Information Theory Creature, mMLC: modified Maximum Likelihood Creature.  \textbf{Left Panel:} the average distance moved per time step as a function of the mobility parameter $x$.  \textbf{Right Panel:} the probability of moving (as opposed to remaining in place) on each time step as a function of the mobility parameter $x$.  
Over the range $1/4<x<1/3$ the movement algorithms strictly correspond to the movement strategies.   Within this range, the MLC moves less frequently and moves a shorter distance on average than the ITC or the Hybrid creature.   However, the mean distance moved by the MLC is only slightly less than that moved by the ITC.  Compare Table 2.5.}
   \label{fig:distance-traveled}
\end{figure}

As Table 2.5 indicates, the average distance moved per time step by a creature following the Max-Likelihood strategy is strictly less than that moved by the Infomax creature, which in turn is less than the mean distance moved by the Hybrid creature.  Similarly, the probability of movement (the probability that $C_{t+1}-C_t\ne 0$) is least under the MLC strategy and greatest under the Hybrid strategy.  Regardless of the detailed nature of the movement cost, the maximum likelihood algorithm would have the greatest advantage in terms of movement efficiency. If resource storage or resource utilization were limiting factors, this might lead to selection against the pure information maximization strategy in favor of the maximum likelihood strategy. 
Likewise, exposure to predation could also impose a selective cost that might further favor the ML creature over the ITC and especially the Hybrid creature.  However it should be noted that as shown in Figure \ref{fig:distance-traveled}  the advantage enjoyed by the MLC in terms of mean distance traveled is very slight compared to the ITC, whereas the absolute difference in the probability of moving is more significant.  Together our results reinforce the conclusion that while information efficiency is surely one important factor in the evolution of adaptive behavior, the significance of information efficiency relative to other considerations is nuanced and context dependent.  

\section{Discussion and Conclusions}\label{sec:conclusion}
\subsection{Summary \& Conclusions}

We developed an analytically tractable model for a particular instance of a perception-action loop: a creature searching for a wandering food source confined to a one-dimensional ring world.  The model encompassed the statistical structure of the possible observations, the effects of the creature's actions on that structure, and the creature's strategic decision making. Although the underlying model took the form of a Markov process on an infinite dimensional state space, we successfully devised an exact coarse graining that reduced the model  to a Markov process on a finite number of ``information states".  This technique allowed us to make quantitative comparisons between the performance of an information-theoretically optimal strategy with other candidate search strategies.  

From the results shown in Figures 2.9 
 and \ref{fig:strategy_food} in Section \ref{sec:asymptotic} we conclude for the simple Distance-Certain model system (Model II) that:
\begin{enumerate}
\item Information optimal search does not necessarily optimize utility.
\item The rank ordering of strategies by information performance does not predict their ordering by expected food obtained, \textit{i.e.}~there is no  simple tradeoff between information and utility.
\item The relative advantage of different strategies depends on the statistical structure of the environment, in particular the variability of motion of the source:
\begin{enumerate}
\item The likelihood maximization strategy outperformed the other two strategies considered (in terms of utility) for a wide range of source mobility.  
\item For the lowest source mobilities considered ($x\gtrsim 1/4$), the hybrid strategy was superior to the other two.
\item The information maximization strategy dominated the other two in a narrow range at the highest source mobility considered ($x\lesssim 1/3$).
\end{enumerate}
\end{enumerate}
Consequently we conclude that the strong hypothesis,   ``behavioral optimality implies information optimality" is false, at least in this model system.  

More broadly, our results suggest that while information efficiency is surely one important factor in the evolution of adaptive behavior, the significance of information efficiency relative to other factors such as expected food gain or expected distance traveled is complex, and needs to be defined through a careful analysis of the organism and its environment.  Rather than proposing to replace the strong hypothesis of information efficiency with an equally sweeping general hypothesis, the model developed here points the way towards a more systematic analysis of the interplay between environmental structure and uncertainties, a creature's needs and means of satisfying them, and the various senses in which it can do so efficiently.  

The difficulties inherent in analyzing systems in which the dependencies of action on sensation and sensation on action form an apparently intractable closed loop are well established.  
In addition to establishing a counterexample against the strong hypothesis, our work also leads to a positive conclusion: that  
despite these difficulties, one may make progress in understanding the interplay of information processing and utility for survival in behavioral systems by embedding a creature and its world together in a combined Markov process.  As we point out in Section \ref{ssec:future} below, this approach can lead to numerous extensions of the present work.


%

\subsection{Limitations of the Model}
\label{ssec:limitations}

How do the simplifications required to obtain a tractable model limit the scope of our conclusions?

The simple model analyzed here incorporates observation, action, and the effects of action on future observations in a nontrivial way to form a closed action-perception cycle.  The key aspect of the model we construct is that the system comprising the ``creature'' and the ``world'' together constitute a discrete state, discrete time Markov chain amenable to analysis using standard probabilistic tools.
One may readily conceive of more elaborate models that include additional elements such as: the cost of memory and computation required to implement a given strategy, imperfect observations, imperfect knowledge of the world, interaction of more than one signal or more than one type of nutrient, constraints such as mortality after excessively long periods without food, goals such as accumulating a sufficient intake of food to achieve reproductive success, and metabolic or predation related movement costs.  

Would any of these additional considerations weaken the conclusion that the model creatures considered here establish a counterexample to the strong hypothesis?  On the contrary.  In the model used here, the information maximization creature has been given all the possible advantages: perfect information about the operating conditions of the world, perfect measurement capability, exact representation of its internal state (probability distribution for the source location), the ability to move arbitrary distances without cost.  If it does not win the competition with all the winds in its favor, it should not stand a chance under more realistic conditions either.  In particular, the addition of memory or computational costs for maintaining the internal state should not affect one creature more than another, as all three strategies are based on the same internal representation of the source probability distribution.  Removing the assumption that the IT creature has perfect knowledge of the world (aside from the source location) and perfect measurement capability would only weaken its performance.  The ITC does not on average move shorter distances or move less frequently than the MLC.  It is likely that more realistic models should reinforce our conclusions.  

We note the significance of the Bayesian update and convolution for predicting the injector's location. Although we assumed a particular geometry for the world (a ring) and a symmetric random walk for the source, our framework would work equally well for a discrete world of any shape where observations (of any kind) are conditionally independent of prior observations given a source's location $S$ and creature's location $C$. For instance, we could consider  2D or 3D worlds, and observations governed by distributions other than Poisson. 
We assumed that the transition matrix $P$ governing movement of the source had translational symmetry, \textit{i.e.}~that the \textit{increments} of the injector's location be independent of the injector's location. However, this assumption could be relaxed without changing the general conclusions of the study. For example, the injector could have landmarks that influence its movement, provided the creature knew their location and influence. In this situation we would still have a Markov process with a transition matrix $P$.  The mapping from PMF for the current location $F^+_t$ to the PMF of the new location $F^-_{t+1}$ would be given by matrix multiplication: $F^-_{t+1}=P \cdot F^+_t$. In this case, the world would not need to be symmetric.

\subsection{Survival Strategies \textit{versus} Movement Algorithms}
\label{ssec:strat-vs-alg}

The movement algorithms are defined in terms of destination given an information state $\ISm{}$.  The algorithms specified in  Figure \ref{fig:strategies} remain well defined over the range $0\le x \le 1/2$.  
However the qualitative properties satisfied by the movement algorithms are only consistent with the given survival strategies (information maximization, likelihood maximization, track and pounce) over narrower ranges of $x$.  The information maximization \emph{strategy} (first maximize information; then maximize expected food among equally informative locations) is consistent with the ITC movement algorithm for all $0<x\le 1/2$.   For $x\equiv 0$, however, the infomax strategy chooses to colocate with the source following the first observation. The movement algorithm for the MLC as defined in Figure \ref{fig:strategies} is consistent with likelihood maximization only when $1/4<x<1/3$.  For $0<x \le 1/4$, the central location is no longer the most likely location in information state $\ISm{1}$.  Consequently a strategy of strict likelihood maximization will no longer remain within information states $\ISm{0}-\ISm{2}$ and the resulting Markov chain will no longer be absorbed by the same positive recurrent set.  The MLC movement \emph{algorithm}, while no longer consistent with strict likelihood maximization outside the range $1/4<x<1/3$, is nevertheless a well defined movement rule and its performance is evaluated in the same fashion as that of the other algorithms in Figure \ref{fig:food-distributions}.  
The hybrid strategy and the hybrid movement algorithm remain consistent over the range $0\le x<1/3$. When $x >1/3$ the MLC, ITC and mMLC strategies coincide with the ITC algorithm.  

\subsection{Relation to Ideal Observer Analysis}

Our analysis of the present model is closely related to Ideal Observer Analysis \cite{GeislerChapter52:2004}. Ideal observer analysis is a framework for quantitative assessment of the task-specific performance of sensory systems  that finds wide use in psychophysics \cite{GreenSwets1966}.  It complements information theory, with which it is sometimes confused; in the theoretical neuroscience literature the two approaches are often assumed to lead to the same conclusions, although there are important examples where they give different results \cite{Thomson+Kristan:2005:NC}.  In the results we show here, both the pursuit of food and the pursuit of information can be construed in terms of ``ideal observers" with differing objectives.  Our results concur with those of \cite{GeislerChapter52:2004,Simoncelli+Olshausen:2001:AnnRevNsci,Thomson+Kristan:2005:NC} and others in emphasizing the importance of task-specific objectives for characterizing ``optimal" behavior of a system; information maximization in the absence of meaningful biological constraints is generally \emph{not} sufficient as an explanatory principle for understanding biological behavior.

\subsection{Future Directions}
\label{ssec:future}

Each of the additional model elements mentioned in Section \ref{ssec:limitations} can be taken into account while preserving the Markovian structure of the models considered.  However, some of them would complicate the model sufficiently to require extensive computational analysis rather than exact treatment along the lines performed here.\footnote{For example the action-perception system studied by Klyubin \textit{et al.}~\cite{KlyubinPolaniNehaniv:2007:NeuralComp} enjoys a Markovian structure, yet required computational rather than direct analysis.}  Whether through analytic or computational means, there is great opportunity for expanding our understanding of the interaction of a creature with a probabilistic environment through embedded Markovian models.  Some directions that are accessible to analysis like that presented here include:
\begin{itemize}
\item Incorporating metabolic thresholds \textit{e.g.}~for starvation (a lower threshold on net metabolic resources) or for reproduction (setting an upper threshold to be reached before the creature can reproduce). Provided the statistical interaction of the creature and its sensorium preserves its Markovian character, the addition of an absorbing boundary allows the system to be analyzed in terms of first passage time distributions, for which there are abundant classical tools available.  Assuming the creature is aware of its own metabolic resources and thresholds and adjusts its strategy accordingly, the creature's internal reserve would become an additional component of its internal state.  
\item Allowing for multiple food types and heterogeneous metabolic resource requirements.  Suppose a creature required two nutrients (``A'' and ``B''), and survival or reproduction depended on avoiding or encountering some boundary in the joint space representing internal reserves of both A and B.  Suppose in addition the creature had limited sensing capacity.  For instance, it might only be able to produce a finite total quantity of cell surface receptors divided between A and B.  Clearly if it were near the threshold for mortality-for-want-of-A, it might choose to reduce the information it took in about B in order to increase its sensitivity to A, despite a net loss in overall information. Within this framework one may study precisely questions about the tradeoffs between net information and ``relevant'' information, a long standing conundrum in the application of information theory to biological systems.  
\item Motile cells and multicellular organisms often navigate by \emph{chemotaxis}, detecting and climbing gradients of signaling molecules carrying information about the location of food or conspecific organisms. 
Chemotaxis based on gradient sensing is an example of a greedy algorithm that locally maximizes the expected rate of gain in the signal. 
As mentioned in the Introduction, Shraiman and colleagues introduced an alternative search strategy, \emph{infotaxis}, that locally maximizes the expected rate of
information gain rather than gain in food or signal concentration in a two dimensional geometry \cite{vergassola2007infotaxis}, and proposed it as an explanation for the zigzagging structure of casting paths observed in the flights of moths pursuing sparsely distributed chemoattractants.  The information maximization strategy we consider here is formally equivalent  to infotaxis.  While we do not focus on the low concentration limit considered in \cite{vergassola2007infotaxis}, it is worth noting that in our system infotaxis outperforms a greedy concentration maximization strategy (the Max-Likelilhood strategy) \emph{only in the case when the source mobility is large, i.e.~the source movement is relatively unpredictable}.  Shraiman and colleagues only considered the case of a stationary source rather than a source performing a random walk.  It would be interesting to  investigate the case of a randomly moving source of a diffusible signal in a two dimensional geometry. 
\end{itemize}
Is there a globally optimal strategy for food gathering for the system considered here?  We have defined a Markovian algorithm as any map from the set of PMFs representing the creature's current information about the source, to the creature's next move. Any such strategy is equivalent to a Markov process on the state space of creature and source locations, observations and PMFs.  Such a process can be described entirely by its transition matrix.  If one considered the transition matrix on information states as a collection of free parameters, one could seek the optimal strategy relative to any objective function (such as maximizing the expected food gained over the long term, or the expected survival duration by some criterion) by numerically or analytically varying the transition matrix.  It would be of great interest to compare the resulting optimal movement algorithm with those defined by heuristic strategies such as the infomax, track-and-pounce, and max likelihood strategies.  As the complications considered above are more fully taken into account it may become possible to generate testable hypotheses about environmental scenarios under which selection pressures might lead to the evolution of distinct strategies, observable in actual organisms, related to the strategies investigated theoretically here.  

\section{Acknowledgments}
PJT acknowledges support of the National Science Foundation, grant DMS-0720142, and research support from the Oberlin College Libraries.  The authors thank K.~Loparo for invaluable advice and discussion and for providing critical comments on a draft of the paper.  The authors also wish to thank M.~Lewicki, E.~Meckes,  B.~Olshausen and T.~Schneider for helpful discussion.  

\appendix


\section{Summary of Notation and Abbreviations}
\label{sec:notation}
Essential notation and abbreviations, in (roughly) order of appearance:

\begin{tabular}{ll}\hline
$N$ & Number of locations in the ring world.\\
$\ZN$ & The ring world, $\{0,1,\cdots,N-1\}$, with addition mod $N$.\\
$t$& Time: $t\in\{0,1,2,\cdots\}$.\\
$S_t$ & Source location at time $t$. \\
$C_t$ & Creature location at time $t$.\\
PMF & Probability Mass Function. \\
$P, P_x$ & PMF describing movement of the source.\\
$x$ & Source mobility parameter, $1/4<x<1/3$.\\
$\pmf$& Set of PMFs on $\ZN$.\\
$\gamma$ & Food injection rate.\\
$\alpha$ & Food decay rate.\\
$y$ & Food molecule transition rate between adjacent nodes (fast timescale).\\
$\lambda_m$ & Mean number of food molecules $m$ steps CW from source.\\
$\Lambda$ & Vector containing mean numbers of food molecules.\\
$Z_t$ & Creature's observation at time $t$: either distance ($D_t$) or number ($M_t$).\\
$\ZZ$ & The space of individual observations, \textit{e.g.}~$M_t\in\mathbb{N}$; $D_t\in\{0,1,\cdots,[[N/2]]\}.$\\
$[[u]]$& The greatest integer less than or equal to a number $u$.\\
$\Hist{t}$ & The creature's history (of locations and observations) as of time $t$.\\
$\allhist$ & The collection of all (infinitely long) histories.\\
$\pmfp{t}$ & PMF of source location at time $t$, given observations through time $t$.  \\
$f^+$ & A particular distribution taken (at random) by some $\pmfp{}$.\\
$\pmfm{t}$ & PMF of source location at time $t$, given observations through time $t-1$.\\
$f^-$ & A particular distribution taken (at random) by some $\pmfm{}$.\\
$\pmfm{0}$ & Initial PMF of source location, before any observations.  $\pmfm{0}\equiv 1/N$.\\
$\ISp{}$ & The set of PMFs for which the injector's location is known, \\
& or is known to be some definite distance $d>0$ from the creature's location.\\
$\ISm{}$ & The set of PMFs obtained from $F\in\ISp{}$ by convolving with $P$.\\
$\ISm{a}$ & Particular information states ($a=0,1,2,3,I$ or $*$).\\
ITC & Information theory Creature: creature following infomax strategy.\\
MLC & Max-Likelihood Creature: creature following likelihood maximization strategy.\\
mMLC & Hybrid or modified Max-Likelihood Creature: creature following hybrid strategy.\\
$\pi$ & A stationary distribution of a Markov process.\\
\hline
\end{tabular}

\section{Food Distribution}\label{app:food}
The distribution of food molecules in Model I is obtained as the steady state of a  rapidly equilibrating discrete time stochastic process. The following food distribution process is assumed to converge to equilibrium much more quickly than the time scale for source or creature movements: (i) Molecules of food enter the world at the location of the food source \textit{via} a Poisson process with mean $\gamma$.  After injection they transition between adjacent nodes with probability $y$.  Molecules leave the world with probability $\alpha$ uniformly in space and time.  We assume both $0<y,\alpha < 1$. 
It can be shown (see \cite{Agarwala2009}, Chapter 3) that 
the molecule counts at each location relative to the injector location are independent and Poisson distributed with means given by a vector $\Lambda$.  We obtain this vector as follows.  Let $Q_y$ and $R_{\gamma}$
be respectively the transition matrix and constant source vector
$$Q_y=\left(\begin{array}{cccc}1-2y & y &&y \\ y&1-2y&& \\&&\ddots&\\y &&&1-2y\end{array}\right)\hspace{1cm}R_{\gamma}=\left(\begin{array}{c}\gamma\\ 0\\ \vdots   \\ 0 \end{array}\right).$$



Then $\Lambda$ satisfies 
$\Lambda=(1-\alpha)Q_y\Lambda+R_{\gamma}.$ Equivalently,
\begin{equation}\label{eq:lambdalemma}
((1-\alpha)Q_y-I){\Lambda}=-R_{\gamma}.
\end{equation}

\begin{lem}\label{lem:lambdaunique}
The matrix $(1-\alpha)Q-I$ is non-singular and hence Equation \ref{eq:lambdalemma} has unique solution $\Lambda$.
\end{lem}
\begin{proof}
$\sum_m Q_{lm}=1$ because every molecule that was just moved must have come from some location with a probability of 1.  
And  $1>1-\alpha>0$ because $1>\alpha>0$.
Therefore, $1>(1-\alpha)\sum_mQ_{lm}>0$.
Subtracting $(1-\alpha)Q_{ll}$ from both sides of $1>(1-\alpha)\sum_mQ_{lm}$ and taking the absolute value results in $|(1-\alpha) Q_{ll}-1|>(1-\alpha)|\sum_{m:l\not=m}Q_{lm}|$. With the triangle inequality $|(1-\alpha) Q_{ll}-1|>(1-\alpha)\sum_{m:l\not=m}|Q_{lm}|$.
Each diagonal entry of the matrix $((1-\alpha)Q-I)$ is $(1-\alpha) Q_{ll}-1$
 and the sum of all entries of any row excluding the diagonal is $(1-\alpha)\sum_{m:l\not=m}|Q_{lm}|$.
Therefore $(1-\alpha)Q-I$ is a diagonally dominant matrix.
 By the Levy-Desplanques theorem \cite{Taussky-1949},
 $(1-\alpha)Q-I$ is non-singular.
\end{proof}
Because $(1-\alpha)Q_y-I$ is non-singular, it is invertible. Therefore,
\begin{align}
\Lambda&=-((1-\alpha)Q_y-I)^{-1}R_{\gamma}. \label{eq:lambdahat}
\end{align}
Considering $\Lambda$ as a function of $\gamma$, 
\begin{align*}
 \Lambda(\gamma) &=-((1-\alpha)Q_y-I)^{-1}R_{\gamma}\\
&=\gamma\left[-((1-\alpha)Q_y-I)^{-1}R_1\right]\\
&=\gamma\Lambda(1).
\end{align*}

Due to the symmetry of $Q_y$ and molecule injection occuring only at the source's location, $\Lambda$ is symmetric about the source location. Therefore the expected molecule counts $m$ steps CW and CCW of the injector are the same.

\section{Information Theory}

Shannon's \textit{Mathematical Theory of Communication}  quantifies the uncertainty of a discrete random variable $X$ taking values $x$ with probability $\Prob(X=x)$ in terms of Boltzmann's \emph{entropy} function (\cite{Shannon1948}, Section I.6)
\begin{align}
H(X)&:=-\left(\sum_{\{x:\Prob(X=x)\not=0\}}\Prob(X=x)\log_2(\Prob(X=x))\right), \label{eq:entropy}
\end{align}
measured here in bits.  
%
%
%
When one variable depends conditionally on another we speak of the
\emph{conditional entropy of $X$ given $Y$}: 
\begin{align}
 H(X|Y)&:=\E_Y(H(X|Y=y))\notag\\
&=\sum_{y}\Prob(Y=y)H(X|Y=y). \label{eq:conditionalentropy}
\end{align}
Observing the variable $Y$ may or may not decrease our uncertainty about $X$,
but it cannot (on average) increase it.  Hence we have the inequality $H(X|Y)\ge H(X)$.
%
%
A quantity of central interest in information theory is the \emph{mutual information} of two random variables.  It quantifies the (average) gain in information of $X$ obtained from an observation of $Y$. Equivalently, it measures the  departure from statistical independence of the random variables $X$ and $Y$ one another.  The mutual information is given by 
\begin{align}
H(X;Y)&:=H(X) - H(X|Y).\label{eq:mutualinfo}
\end{align}
In Section \ref{sec:results} we apply this expression to measure the creature's information about the source location.  In that case we take $H(X)$ to be the distribution of the source unconditioned on any observations, namely the uniform distribution on the ring, with entropy of $\log_2(N)$ bits.    We take the conditional entropy $H(X|Y)$ to be the creature's uncertainty about the source's next location after making an observation, \textit{i.e.}~the distribution $\pmfm{}$, defined in Appendix \ref{app:F}. 


\section{Histories and Internal States}\label{app:histories}
\subsection{Histories}

The creature searching for the source knows only the history of its own locations (relative to its starting position) and its observations of the local food concentration.  In Model I these observations take the form of molecule counts at the creature's location; in Model II the creature observes the inferred distance to the food source.  Let the random variable $Z_t$ represent the observation at time $t$ in either case. 
 Before making the observation, the creature knows its location but not how many molecules are present, or how distant the source is.  

Technically, each history is one of an infinite set of all possible histories, a typical element of which would be 
$$h=\{C_0=c_0,Z_0=z_0,C_1=c_1,Z_1=z_1,\cdots,C_t=c_t,Z_t=z_t,\cdots\}$$
(recall, however that we set $C_0\equiv 0$ with probability one). 
When we consider the information known to the creature up to a certain point we will use conditional probabilities, conditioning on the set of histories specified up to the given time.\footnote{Our usage parallels the language of a filtration, an increasing family of $\sigma$-algebras in the theory of continuous time stochastic processes.  The associated measure theoretic machinery will not be needed for the arguments here and we will simplify the notation as much as possible.}  For the argument in the main text it suffices to consider histories as lists of observations-to-date, but for completeness we give here a more formal definition, which is necessary to state the definition of the random measures $F^{\pm}_t$ precisely. 

Let $\allhist$ be the set of all histories of the form $h$ given above, and $\ZZ$ the space of individual observations ($\ZZ=\mathbb{N}$ in Model I, and $\ZZ=\{0,1,\cdots,[[N/2]]\}$ in Model II).  We define projection operators $\pi_t^c$ and $\pi_t^z$ from $\Hist{}$ to $\mathbb{Z}_N$ and $\ZZ$, respectively, as
\begin{eqnarray*}
\pi_t^c:h\in\allhist&\longrightarrow&c_t\\
\pi_t^z:h\in\allhist&\longrightarrow&z_t.\\
\end{eqnarray*}
We define \textit{the history at time $t$} as the (random) subset of all histories consistent with the observations through time $t$: 
$$\Hist{t}=\{h\in\allhist | \pi_0^c(h)=C_0,\pi_0^z(h)=Z_0,\cdots,\pi_t^c(h)=C_t,\pi_t^z(h)=Z_t\}.$$
On each time step $t$ the latest creature location and observation is added \textit{e.g.}~:
\begin{eqnarray*}
\Hist{0}&=&\{h\in\allhist | \pi_0^c(h)=C_0,\pi_0^z(h)=Z_0\}\\
\Hist{1}&=&\{h\in\allhist | \pi_0^c(h)=C_0,\pi_0^z(h)=Z_0,\pi_1^c(h)=C_1,\pi_1^z(h)=Z_1\}\\
&\vdots&
\end{eqnarray*}
Consequently the histories form a system of nested subsets of $\allhist$:
$$\allhist\supset\Hist{0}\supset\Hist{1}\supset\cdots\supset\Hist{t-1}\supset\Hist{t}\supset\Hist{t+1}\supset\cdots$$
and we can form successive histories by intersection
\begin{eqnarray}\label{eq:history-intersection}
\Hist{t}&=&\Hist{t-1}\bigcap\{h\in\allhist|\pi^c_t(h)=C_t,\pi^z_t(h)=Z_t\}.
\end{eqnarray}

\subsection{Internal States, or the Random Measures $F^{\pm}$}\label{app:F}
Given a history corresponding to the random sequences $\{C_t\}$ and $\{Z_t\}$, we can define what an ideal observer would be able to infer about the location of the source given the information available to the creature in terms of two probability mass functions.  One ($\pmfp{t}$, say) captures what a creature knows about the current location of the source upon making the observation $Z_t$ at time $t$.  The other ($\pmfm{t+1}$, say) captures what a creature knows about the \emph{next} location of the source, before the creature itself moves and makes another observation.  Hence $\pmfm{t+1}$ is determined entirely by $\pmfp{t}$ and the rule governing the movement of the source.  Formally, we may write
\begin{eqnarray}
\pmfm{t}(l)&=&\Prob(S_t=l|\Hist{t-1})\\
\pmfp{t}(l)&=&\Prob(S_t=l|\Hist{t}) \label{eq:pmfp-defined}\label{eq:pmfp-def}
\end{eqnarray}
for $0\le l \le N-1$.  Hence $F^{\pm}_t$ are random vectors that are also measures, or elements of the probability simplex on the ring
$$\pmf=\{f\in\mathbb{R}^N|f\ge 0, \sum_i f_i=1\}.$$

\section{Bayesian Update of Probability Mass Functions}
\label{app:bayes}
Here we derive Equation \ref{eq:pmfp} giving the Bayesian rule that updates the probability mass function $\pmfp{t}$, which is the PMF of the source's location at time $t$ given the creature's history. 
We start from the definition of $\pmfp{t}$ given by Equation \ref{eq:pmfp-defined}.
For the conditional probability $P(S_t=l|h\in\Hist{t})$ to be defined
we restrict consideration to histories $\Hist{t}$ that have strictly positive possibility of occurring, \textit{i.e.}~we restrict $\Hist{t}$ such that $\Prob(h\in\Hist{t})>0$.  To consolidate notation, let $\ip$ denote the event $S_t=l$, $\cp$ denote the event that $C_t=c$, $\obs$ denote the event that $Z_t=z$, and $\Histg{t}$ denote the event $h\in\Hist{t}$. 
In this notation, note that $\Histg{t}=\obs\+\cp\+\Histg{t-1}$ (\textit{cf.}~Equation \ref{eq:history-intersection}). Applying Bayes law, $\Prob(A\+ B)=\Prob(A|B) \Prob(B)$, Equation \ref{eq:pmfp-def} for $\pmfp{t}(s)$ can be expanded as follows:
\begin{align}
\pmfp{t}(l)&=\Prob(\ip|\Histg{t})\notag\\
&=\frac{\Prob(\ip\+ \obs\+\cp\+ \Histg{t-1})}{\Prob(\obs\+\cp\+ \Histg{t-1})}\label{eq:St|h2}\\
&=\frac{\Prob(\obs| \cp\+\ip\+  \Histg{t-1})\cdot \Prob(\ip| \cp\+\Histg{t-1})\cdot \Prob(\cp\+\Histg{t-1})}{\Prob(\obs|\cp\+\Histg{t-1})\cdot \Prob(\cp\+\Histg{t-1})} \label{eq:St|h3}\\
&=\frac{\Prob(\obs| \cp\+\ip\+  \Histg{t-1})\cdot \Prob(\ip| \cp\+\Histg{t-1})}{\Prob(\obs|\cp\+\Histg{t-1})}. \label{eq:St|h}
\end{align}

The number of molecules at $C_t$ is a Poisson random variable with mean $\lambda_{S_t-C_t}$ and the distance $D_t$ is $S_t-C_t$, therefore both observations are conditionally independent of the history given $C_t=c$ and $S_t=l$.
But because  the creature's observation  $Z_t$ is conditionally independent of the history given $C_t=c$ and $S_t=l$, we have
 \begin{align*}
\Prob(\obs| \cp\+\ip\+ \Histg{t-1} ) &= \Prob(\obs| \cp\+\ip).
\end{align*}
The injector's location is conditionally independent of the creature's location $C_t$ given $\Hist{t-1}$ because $S_t=S_{t-1}*P$ and $C_t$ is a mapping from $\Hist{t-1}$ to $\ZN$; this fact implies
\begin{align*}
\Prob(\ip| \cp\+\Histg{t-1})&=\Prob(\ip|\Histg{t-1})\\
&=\pmfm{t}(l).
\end{align*}
The second equality follows from the definition of $\pmfm{t}$ in Appendix \ref{app:F}. 
Applying Bayes law and both examples of conditional independence to the term $\Prob(\obs|\cp\+\Histg{t-1})$ we see,
 \begin{align*}
\Prob(\obs| \cp\+ \Histg{t-1})&= \sum_{m=0}^{N-1}\Prob(\obs|S_t=m\+ \cp\+\Histg{t-1})\cdot\Prob(S_t=m|\cp\+\Histg{t-1})\\
&=\sum_{m=0}^{N-1}\Prob(\obs|S_t=m\+ \cp)\cdot\Prob(S_t=m|\Histg{t-1})\\
&=\sum_{m=0}^{N-1}\Prob(\obs|S_t=m\+ \cp)\cdot\pmfm{t}(m).
\end{align*}
Therefore, $\pmfp{t}(l)$ can be written as
\begin{align*}
\pmfp{t}(l)&=\frac{\Prob(\obs| \cp\+\ip\+  \Histg{t-1})\cdot \Prob(\ip| \cp\+\Histg{t-1})}{\Prob(\obs|\cp\+\Histg{t-1})}\\
&=\frac{\Prob(\obs| \cp\+\ip)\cdot \pmfm{t}(l)}{\sum_{m=0}^{N-1}\Prob(\obs|S_t=m\+ \cp)\cdot\pmfm{t}(m)},
\end{align*}
which concludes the derivation of Equation \ref{eq:pmfp}.

\section{Strategies}\label{app:strategies}

\subsection{Strategies Defined in Terms of Probability Mass Functions and Information States}

The Max-Likelihood strategy is defined as the strategy under which the creature  moves to the injector's \emph{most likely next location}, given the creature's history.  If more than one location is equally likely, the MLC algorithm chooses the closest.  If two are equally likely and equally close, the algorithm chooses between them at random.  Given the next source location PMF $\pmfm{t+1}$, define two subsets of locations $\mathcal{E}'_M\subset \mathcal{E}_M\subset\ZN$ as
\begin{eqnarray}
\mathcal{E}_M&=&\argmax{l\in\ZN}\pmfm{t+1}(l)\\
\mathcal{E}'_M&=&\argmin{l'\in \mathcal{E}'_M}|C_t-l'|
\end{eqnarray}
and choose $C_{t+1}|\pmfm{t+1}$ at random, uniformly, from ${\mathcal E}'_M$.  

Under the Infomax strategy the creature moves to a location that will give the lowest expected uncertainty of the source's location, upon the next observation.  When more than one location would satisfy this criterion, the creature selects among the most informative locations those that provide the highest expected food reward.  If multiple locations satisfy \emph{both} the information maximization and utility maximization criteria, the creature selects the closest such location.  If two  locations  are equally informative, food-rich, and close, the creature selects between them at random.  Given the source location PMF $\pmfm{t+1}$, define sets of locations $\mathcal{E}''_I\subset \mathcal{E}_I'\subset\mathcal{E}_I\subset\ZN$ as
\begin{eqnarray}
{\mathcal E}_I&=&\argmax{l\in\ZN}\E[-H(\pmfm{t+2})|C_{t+1}=l]\\
{\mathcal E}'_I&=&\argmax{l'\in{\mathcal E}_I} \E(M_{t+1}|C_{t+1}=l')\\
{\mathcal E}''_I&=&\argmin{l''\in{\mathcal E}'_I}|C_t-l'' |
\end{eqnarray}
Here $H(f)$ is the entropy of the PMF $f\in\pmf$.  
Choose $C_{t+1}|\pmfm{t+1}$ at random, uniformly, from ${\mathcal E}''_I$.

The modified Maximum Likelihood strategy (\textit{a.k.a.}~the Hybrid or the Track-and-Pounce strategy) chooses between the two previous strategies, depending on the creature's information state.  When $\pmfm{t+1}\in\ISm{0}$, the creature selects $C_{t+1}$ from $\mathcal{E}'_M$, otherwise the creature selects $C_{t+1}$ from $\mathcal{E}''_I$.

\subsection{The Infomax Strategy Minimizes Uncertainty}\label{app:infomaxstrat}


The mutual information between the source's location at time $t$ and the creature's prior history is determined by Equation \ref{eq:mutualinfo}. This is equivalent to applying the same equation to each information state. From Lemma \ref{lem:probS}, we have the first term of Equation \ref{eq:mutualinfo} and we then calculate the conditional entropy of each State.

\begin{lem}\label{lem:probS}
The (unconditioned) probability that the injector is at location $s$ at time $t$ is $1/N$. That is, $\Prob(S^t=s)=1/N$.
\end{lem}
\begin{proof} 
The transition matrix $P$ generating the injector's next location is independent of its current location.  In addition, the initial distribution of the injector is assumed to be uniform on the ring, \textit{i.e.}~$\forall s$ $\Prob(S_0=s)=1/N$.
Consider the injector starting at some location, $l$, moving along some path according to $P$ and ending up at location $m$. 
It is equally probable that the injector
 could have started at $l+a$ because $$\Prob(S_0=l)=\Prob(S_0=l+a)=1/N,$$
 moved along the same (equally probable) path, and
ended up at location $m+a$.
This argument holds for all starting locations and paths so $\forall a,\forall t$ $$\Prob(S^t=s)=\Prob(S^t=s+a).$$ Therefore, $\forall t,\forall s,\, \Prob(S^t=s)=1/N$.
\end{proof}

The infomax strategy greedily optimizes for expected information and then expected food. For the Distance-Certain Model we propose a strategy in Figure \ref{fig:strategies} which we now show is an infomax strategy. Furthermore, we demonstrate that the infomax strategy is long-term optimal for information. First we demonstrate that a creature in $\ISm{0}$ has more information about the injector's location than a creature in any other state and then point out that the strategy proposed in Figure \ref{fig:strategies} always transitions to $\ISm{0}$. Finally, we show that the proposed strategy selects the location with highest expected food that is guaranteed to transition to $\ISm{0}$.

\begin{lem}\label{lem:ISinfo}
A creature in the information state $\ISm{0}$ has optimal information about the injector's location.
\end{lem}
\begin{proof} Recall that $\ISp{}$ is the set of information states in which the distance to the source is known with certainty.  
 We partition $\ISp{}$ into the following subsets: $\ISp{0}$ where the injector's location is known, $I^{1+}$ where the injector is known to be one step away from some location but not in $\ISp{0}$, and $I^{2+}$ where the injector is known to be a distance $d$ away from some location ($d\ge 2$) but not in $\ISp{0}$.  This partition of $\ISp{}$ leads to a natural partition of $\ISm{}$ where $\ISm{0}$, $I^{1-}$, and $I^{2-}$ are generated from $\ISp{0}$, $I^{1+}$, and $I^{2+}$ respectively by application of the convolution representing the injector's movement algorithm $P_x$.  Note that $\ISm{1}\subset I^{1-}$ because $I^{1-}$ also includes situations where the source is not equally likely to be CW as CCW.
 
Each $f^0$ in $\ISm{0}$ has the form
\begin{align}
	f^0(l)&=
\left\lbrace
\begin{array}{ll}
	x 	& 	l=e-1\\
	1-2x	& 	l=e\\
	x 	& 	l=e+1\\
	0 	& 	\text{otherwise}
\end{array}
\right.
\end{align}
with $e\in\{0,1,\cdots,N-1\}$.

Each $f^1$ in $I^{1-}$ has the form
\begin{align}
	f^1(l)&=
\left\lbrace
\begin{array}{ll}
	px 		& 	l=e-2\\
	p(1-2x) 	& 	l=e-1\\
	x		& 	l=e\\
	q(1-2x)	& 	l=e+1\\
	qx 	& 	l=e+2\\
	0 		& 	\text{otherwise}
\end{array}
\right.\label{eq:f-1}
\end{align}
with $e\in\{0,1,\cdots,N-1\}$, $p,q>0$, and $p+q=1$. Here $p$ is the probability that the injector was 1 step CCW before moving.

Each $f^2$ in $I^{2-}$ has the form
\begin{align}
	f^2(l)&=
\left\lbrace
\begin{array}{ll}
	px 		& 	l=e-d\pm1\\
	p(1-2x)	& 	l=e-d\\
	qx 	& 	l=e+d\pm1\\
	q(1-2x)	& 	l=e+d\\
	0 		& 	\text{otherwise}
\end{array}
\right.\label{eq:f-2}
\end{align}
with $e\in\{0,1,\cdots,N-1\}$, $d\in\{2,3,\cdots,N-2\}$, $p,q>0$, and $p+q=1$. Here $p$ is the probability that the injector was $1<d<N-1$ steps CCW before moving.

If a creature has a history that maps $\pmfm{}$ into $\ISm{0}$, then the mutual information between the injector's location and the creature's prior history is given by Equation \ref{eq:mutualinfo} as
\begin{align*}
 H_0(S_t;h\in\Hist{t-1}\cap h\mapsto\ISm{0})&=H(S^t) - H(S^t|h\in\Hist{t-1}\cap h\mapsto\ISm{0})\\
&=H(S^t)-\E_{f^0\in\ISm{0}} [H(\pmfm{t}=f^0)]\\
&=\log_2 (N) +2x\log_2(x)+(1-2x)\log_2(1-2x).
\end{align*}
 For convenience, let us define the function $\phi:[0,1]\mapsto [0,1/2]$ as 
 \begin{equation}
 \phi(x)=\left\{\begin{array}{ll}
 0,&x=0\\
 x\log_2(x),&0<x\le 1.
 \end{array}\right.
 \end{equation}
If a creature has a history that maps $\pmfm{}$ into $I^{1-}$, then the mutual information between the injector's location and the creature's prior history (given by Equation \ref{eq:mutualinfo}) depends on the CCW probability $p, 0<p<1$. (If $p=0$ or $p=1$ then we are back in information state $\ISp{0}$, as the location of the source is known with certainty).  As a function of $p$, the mutual information is 
\begin{align*}
H_1(S_t;h\in\Hist{t-1}\cap h\mapsto I^{1-})=&H(S^t) - H(S^t|h\in\Hist{t-1}\cap h\mapsto I^{1-})\\
=&\log_2(N)+2x\log_2(x)+(1-2x)\log_2(1-2x)+\phi(p)(1-x)\\
=&H_0(S_t;h\in\Hist{t-1}\cap h\mapsto\ISm{0})+\phi(p)(1-x).
\end{align*}
Similarly, if a creature has a history that maps $\pmfm{}$ into $I^{2-}$, then the mutual information between the injector's location and the creature's prior history is
\begin{align*}
H_2(S_t;h\in\Hist{t-1}\cap h\mapsto I^{2-})=&H(S^t) - H(S^t|h\in\Hist{t-1}\cap h\mapsto I^{2-})\\
=&\log_2(N)+2x\log_2(x)+(1-2x)\log_2(1-2x)+\phi(p)\\
=&H_0(S_t;h\in\Hist{t-1}\cap h\mapsto\ISm{0})+\phi(p).
\end{align*}
As we have restricted the source mobility parameter $x$ to the range $1/4\le x \le 1/3$, we have the strict inequality that for each $p$, $H_0>\max(H_1,H_2).$
That is, the mutual information (the reduction in entropy of $S_t$ upon conditioning on the creature's history $\Hist{t-1}$) is greatest for a creature in information state $\ISm{0}$.  
As $\ISm{0}, I^{1-}$, and $I^{2-}$ partition all possible histories in the Distance Certain Model (together with the initial state $f_{u}$ which represents a state of having no information about the source), a creature in the information state $\ISm{0}$ has the maximal possible  information about the injector's location.\qedhere
\end{proof}

\begin{cor}\label{lem:infomaxstrat}
The Infomax Strategy has (long-term) optimal information about the injector's location.
\end{cor}
\begin{proof} 
From Lemma \ref{lem:ISinfo}, the information state $\ISm{0}$ minimizes the creature's uncertainty about the source's location, \textit{i.e.}~it minimizes the entropy of $f\in\ISm{}$. The infomax strategy illustrated  
in Figure \ref{fig:strategies} transitions to $\ISm{0}$ with probability one (see Figure \ref{fig:strategy_transitions}). Therefore, with probability one, the infomax strategy has optimal information about the injector's location. \qedhere
\end{proof}

The infomax strategy must choose between locations that have optimal expected information on the next turn. Lemma \ref{lem:ISinfo} shows that a location that is guaranteed to transition to $\ISm{0}$ has the highest possible expected information. The infomax strategy proposed in Figure \ref{fig:strategies} chooses the location that is guaranteed to transition to $\ISm{0}$ and has the highest expected food of all locations that transition to $\ISm{0}$. Therefore Figure \ref{fig:strategies} depicts an infomax strategy.

\subsubsection{Mutual Information for Sequence Entropies}
\label{sssec:MI-sequence}
In addition to calculating the mutual information as the expected reduction in entropy between $f_u\equiv 1/N$ and $\pmfm{}$, one may consider the differences in entropies between \emph{ensembles of random sequences} of source locations on  the one hand, and creature locations and internal states on the other.  One may compare the  entropy of the finite sequence $(S_0,S_1,\cdots,S_T)$ unconditioned on any observations
with the mean entropy after conditioning on the creature's trajectory $(C_0,C_1,\cdots,C_T)$ and its sequence of internal states $(\pmfp{0},\pmfp{1},\cdots,\pmfp{T})$.  Allowing for a slight abuse of notation, let the (negative) entropy function $\phi$ defined above be written with three arguments in the case of three possible outcomes: we will write $\phi(a,b,c)$ for $a\log_2 a+b\log_2 b + c\log_2 c$, for example, or 
$\phi(x,1-2x,x)$ for $2x\log_2(x)+(1-2x)\log_2(1-2x)$.  It is easy to show, using chains of conditional probabilities, that the entropy of the finite sequence of source locations $(S_0,S_1,\cdots,S_T)$  is given by
\begin{equation}
H[\{S_t\}_{t=0}^T] = \log_2(N)-T\phi(x,1-2x,x)
\end{equation}
where the $\log_2(N)$ term reflects the entropy due to the initial (uniform) distribution and the $\phi(x,1-2x,x)$ terms reflect the contribution to the entropy of the finite sequence on each successive time step.\footnote{We note that this construction works precisely because the source alone -- like the entire system -- is  Markov process.}
We may define the \emph{sequence entropy} to be the mean growth rate of the entropy of the sequence over time,\footnote{We note that Schreiber introduced a related notion, the \emph{transfer entropy}, to quantify the transfer of information from one sequence of random variables to another \cite{Schreiber:2000:PRL}.}
 \textit{i.e.}~
\begin{equation}
H_{seq}=\lim_{T\to\infty}\left( \frac{1}{T} H[\{S_t\}_{t=0}^T]\right)=\lim_{T\to\infty} \left( \frac{1}{T} \left[\log_2(N)-T\phi(x,1-2x,x)\right]\right)=-\phi(x,1-2x,x).
\end{equation}
As expected for any Markov process that tends towards an equilibrium distribution, the entropy of the sequence $T$ steps after the initial condition asymptotically grows linearly in time.  One may carry out a similar calculation for the entire system restricted to the set of one, two or three absorbing states (see \textit{e.g.}~Appendix \ref{app:strategytransitionmatrices}).  It suffices to point out, however,
that if the creature adopts the infomax strategy, hence returns on each step to the ``most informed'' of the information states $\ISp{0}$, then the \emph{sequence of source locations conditioned on the creature's history} has asymptotic entropy of zero.  The creature that remains in the information state $\ISp{0}$ can retrospectively reconstruct the source trajectory perfectly.  The other creatures cannot, and hence the infomax strategy is information-optimal both from the point of view of the next-step entropies and of the entropy of the entire sequence of source trajectories and creature histories.

\subsection{Transition Matrices}\label{app:strategytransitionmatrices}

The state space of any finite Markov chain may be uniquely decomposed into a set of transient states and one or more irreducible sets of recurrent states.  A state is \emph{recurrent} if the probability of eventual return to that state is unity, otherwise the state is \emph{transient}.  A set of states is \emph{irreducible} if all the states in the set intercommunicate, \textit{i.e.}~there is a positive probability of moving between any pair over a finite number of iterations.  For example, if $M$ is a transition matrix defined on a set of states of the Markov chain then the set is irreducible if every entry of $M^k$ is positive for some $k$.  An irreducible Markov chain on a finite state space will have a unique stationary distribution, $\mathbf{\pi}$ \cite{Grimmett+Stirzaker:2005:book}.  We will show that each of the strategies considered reduces to an irreducible Markov chain on a finite set of recurrent states.  We call this subset of the original Markov chain the \emph{absorbing set}.  We will calculate the stationary distribution $\mathbf{\pi}$ in each case.

The transition matrix between information states for the Infomax strategy is trivial.  The system enters the state $\ISm{0}$ by the second iteration and remains in state $\ISm{0}$ from then on, with probability 1. In Table \ref{tab:5x5matrices}, the matrix $M_{ITC}$ gives the transitions restricted to the five relevant information states.  State $\ISm{0}$ corresponds to the second column, and is absorbing after two iterations for all initial conditions.  Equivalently, the matrix $M_{ITC}^2$ consists of a single nonzero column containing all unit entries.  

The Hybrid strategy has an absorbing set consisting of the states $\ISm{0}$ and $\ISm{1}$. The transition matrix for the strategy restricted to these states is
\begin{align}
M&=\left[
\begin{array}{ll}
 1-2x & 2x\\
1 & 0
\end{array}
\right].
\end{align}
This set of states is irreducible and recurrent since $M^2$ has all positive entries in this case.  The stationary distribution is
\renewcommand{\arraystretch}{1.5}
\begin{align}
 \pi_{mMLC} &= \left[
\begin{array}{l}
\frac{1}{1+2x}\\
\frac{2x}{1+2x}
\end{array}
\right].
\end{align}
\renewcommand{\arraystretch}{1}

The Max-Likelihood strategy has an absorbing set consisting of the states $\ISm{0}, \ISm{1}$ and $\ISm{2}$.  The transition matrix restricted to these states is
\begin{align}
M&=\left[
\begin{array}{lll}
 1-2x & 2x & 0\\
x & 1-2x & x\\
1-x & x & 0
\end{array}
\right].
\end{align}
In this case as well these three sets for an irreducible recurrent chain; the matrix $M^2$ again has all positive entries, for the range of $x$ considered ($1/4<x<1/3$).
The stationary distribution is 
\renewcommand{\arraystretch}{1.5}
\begin{align}
 \pi_{MLC} &=\left[
\begin{array}{l}
\frac{2-x}{4+x}\\
\frac{2}{4+x}\\
\frac{2x}{4+x}
\end{array}
\right].
\end{align}
\renewcommand{\arraystretch}{1}

\section{The Expected Food and Information}\label{app:ergodic}

The mean food intake is the expected value (under the asymptotic probability distribution for the Markov process on the absorbing set) of the molecule count.  Since the system is ergodic, this expected value is equal (``almost surely") to the long time average of the molecule count.  That is
$$\lim_{T\to\infty}\frac{1}{T}\sum_{t=1}^T M_t = \E[M_t]\hspace{1cm}\mbox{(almost surely)}.$$
For a given strategy, the expected food is then the dot product of the stationary distribution $\pi$ and the expected food for each state.
Similarly, the expected information is the dot product of the stationary distribution $\pi$ and the information in each state.  


Let $\phi(x)$ denote the (negative) entropy function $\phi(x)=2x\log_2(x)+(1-2x)\log_2(1-2x)$ as in Appendix \ref{sssec:MI-sequence}.  

Table \ref{tab:entropies} shows the entropies of the probability mass functions $\pmfm{}$ representing the uncertainty of the injector's next location in each information state.  If $H_k(x)$ denotes the entropy of the PMF for state $\ISm{k}$ (as a function of the injector's mobility, $x$) then we have $H_0(x)<H_1(x)<H_2(x)=H_k(x), k\ge 2$.  
\begin{table}
\centering
\label{tab:entropies}
\begin{tabular}{cl}
Information State & Entropy of PMF $\pmfm{}$\\
$\ISm{0}$&$H_0(x) = -\phi(x)$\\
$\ISm{1}$&$H_1(x)=  -\phi(x)+1-x$\\
$\ISm{2}$&$H_2(x)=  -\phi(x) + 1 $
\end{tabular}
\caption[Entropies of PMFs by information state.]{Entropy of probability mass functions (PMFs) in different information states.  The PMF representing the next location of the source when the creature is in information state $k\in\{0,1,2\}$, is given as a function of the source mobility $x$ by $H_k(x)$.  We restrict the source mobility to $1/4<x<1/3$.  For $k>2$ we have $H_k(x)=H_2(x)$.  Clearly $H_0(x)<H_1(x)<H_2(x)$.  Notation: for $0<x<1/2$, $\phi(x)$ is here defined as 
$\phi(x)=2x\log_2(x)+(1-2x)\log_2(1-2x)$.}
\end{table}

The long-run average entropy of the information states for each strategy is obtained from Table \ref{tab:entropies} and the equilibrium distributions in Appendix \ref{app:strategytransitionmatrices}. Let $\bar{H}_{ITC}$, $\bar{H}_{mMLC}$ and $\bar{H}_MLC$ be the mean entropies of the internal state $\pmfm{}$ for the infomax, hybrid and max-likelihood creatures, respectively.  From a straightforward calculation we obtain:
\begin{eqnarray*}
\bar{H}_{ITC}&=&-\phi(x)\\
\bar{H}_{mMLC}&=&-\phi(x)+2x(1-x)/(1+2x)\\
\bar{H}_{MLC}&=&-\phi(x)+2/(4+x),
\end{eqnarray*}
from which it follows (in the range $1/4\le x \le 1/3$) that the infomax strategy gives the source location estimates with the lowest mean entropy:
$$\bar{H}_{ITC}<\bar{H}_{mMLC}<\bar{H}_{MLC}.$$

%

\section{Joint (Distance, Information State) Transitions}
\label{app:joint}

To understand the dependence of expected food gained on source mobility, we derive the \textit{joint} transition probabilities between states representing both the distance $D_t$ between the creature and the source at each time step and the creature's information state $\ISm{}_t$, for the information states that form the positive recurrent states for each strategy.  
Let the notation $(n,k)$ denote the creature at time $t$ being a distance $D_t=n$ from the food source and in information state $\ISm{k}_t$.  
The conditional probability $\Prob_S(D_t=n|\ISm{k})$ that a creature following strategy $S$ ends up  a distance $n$ from the source, given that the creature is in information state $\ISm{k}$, is as follows:
\begin{eqnarray}
\Prob_{ITC}(D_t=n|\ISm{k}_t)&=&\begin{array}{c|c|}
(n,k)	& (\cdot,0) \\ \hline
(0,\cdot) & x\\
(1,\cdot) & 1-2x\\
(2,\cdot) & x\\\hline
                \end{array}
\\
\Prob_{mMLC}(D_t=n|\ISm{k}_t)&=&\begin{array}{c|cc|}
(n,k)&(\cdot,0)&(\cdot,1)\\ \hline
(0,\cdot)&1-2x	&x/2\\
(1,\cdot)&2x	&(1-2x)/2\\
(2,\cdot)&0	&x\\
(3,\cdot)&0	&(1-2x)/2\\
(4,\cdot)&0	&x/2\\\hline
\end{array}
\\
\Prob_{MLC}(D_t=n|\ISm{k}_t)&=&\begin{array}{c|ccc|}
(n,k)&(\cdot,0)&(\cdot,1)&(\cdot,2)\\ \hline
(0,\cdot)&1-2x	&x	&(1-2x)/2\\
(1,\cdot)&2x	&1-2x	&x\\
(2,\cdot)&0	&x	&0\\
(3,\cdot)&0 	&0	&x/2\\
(4,\cdot)&0	&0	&(1-2x)/2\\
(5,\cdot)&0	&0	&x/2\\ \hline 
\end{array}.
\end{eqnarray}

To find the probability of being distance $n$ from the injector we apply Bayes law, obtaining $\Prob_S(D=n)=\sum_{k}\Prob_S(D=n|\ISm{k})\Prob_S(\pmfm{}\in\ISm{k})$.  We sum the index $k$  over the positive recurrent information states for strategy $S$. The probability of being in each information state for each creature is calculated in Appendix \ref{app:strategytransitionmatrices},
and the resulting probabilities for each combination of mobility and strategy are given in Table 2.3. 

With rows correspond to states at time $t$ and columns correspond to states at time $t+1$, the joint (distance, information) transition matrices on the positive recurrent states are:
\begin{eqnarray}\label{eq:joint-trans-ITC}
T_{ITC}&=&\begin{array}{c|ccc|}
(n,k)	&(0,0)	&(1,0)	&(2,0)	\\ \hline
(0,0)	&x	&1-2x	&x\\
(1,0)	&x	&1-2x	&x\\
(2,0)	&x	&1-2x	&x\\ \hline 
\end{array}
\\ \label{eq:joint-trans-mMLC}
T_{mMLC}&=&\begin{array}{c|ccccccc|}
(n,k)	&(0,0)	&(1,0)	&(0,1)	&(1,1)	&(1,2)	&(1,3)	&(1,4)\\ \hline
(0,0)	&1-2 x & 2 x & 0 & 0 & 0 & 0 & 0 \\
(1,0)	& 0 & 0 & \frac{x}{2} & \frac{1}{2} (1-2 x) & x & \frac{1}{2} (1-2 x) & \frac{x}{2} \\
(0,1)	& 1-2 x & 2 x & 0 & 0 & 0 & 0 & 0 \\
(1,1)	& 1-2 x & 2 x & 0 & 0 & 0 & 0 & 0 \\
(2,1)	& 1-2 x & 2 x & 0 & 0 & 0 & 0 & 0 \\
(3,1)	& 1-2 x & 2 x & 0 & 0 & 0 & 0 & 0 \\
(4,1)	& 1-2 x & 2 x & 0 & 0 & 0 & 0 & 0 \\ \hline
\end{array}
\\ \label{eq:joint-trans-MLC}
\\ \nonumber
T_{MLC}&=&\begin{array}{c|cccccccccc|}
(n,k)	&(0,0)	&(1,0)	&(0,1)	&(1,1)	&(2,1)	&(0,2)		&(1,2)	&(3,2)	&(4,2)		&(5,2)	\\ \hline
(0,0)	&1-2x	&2x	&0	&0	&0	&0		&0	&0	&0		&0	\\
(1,0)	&0	&0	&x	&1-2x	&x	&0		&0	&0	&0		&0	\\
(0,1)	&1-2x	&2x	&0	&0	&0	&0		&0	&0	&0		&0	\\
(1,1)	&0	&0	&x	&1-2x	&x	&0		&0	&0	&0		&0	\\
(2,1)	&0	&0	&0	&0	&0	&\frac{1-2x}{2}	&x	&\frac{x}{2}	&\frac{1-2x}{2}	&\frac{x}{2}	\\
(0,2)	&1-2x	&2x	&0	&0	&0	&0		&0	&0	&0		&0	\\
(1,2)	&0	&0	&x	&1-2x	&x	&0		&0	&0	&0		&0	\\
(3,2)	&1-2x	&2x	&0	&0	&0	&0		&0	&0	&0		&0	\\
(4,2)	&1-2x	&2x	&0	&0	&0	&0		&0	&0	&0		&0	\\
(5,2)	&1-2x	&2x	&0	&0	&0	&0		&0	&0	&0		&0	\\ \hline 
\end{array}.
\end{eqnarray}

From these transition matrices we obtain the stationary probability vectors for the joint (information state, distance) distributions, from which one immediately obtains the fraction of time spent at each distance relative to the food source.  In the case of the MLC each distance $0\le D \le 5$ appears in exactly one joint state, so the matrix $T_{MLC}$ may be viewed as defining a Markov process on the six recurring distances zero through five.  In the case of the hybrid creature the distance zero occurs both in conjunction with the information state 0 and the information state 1, which prevents representation of this strategy's behavior as a Markov process solely in terms of distances.  
The steady state distributions for the \emph{distances} for the hybrid and MLC strategies are given algebraically in Table 2.3, 
 and graphically in Figure \ref{fig:distance-probs}.

\section{Expected Distance Traveled}\label{app:distance-traveled}

We illustrate the calculation of expected distance traveled and movement probability under the infomax algorithm.  The other cases are similar, requiring only additional bookkeeping.  
Consider the transitions for the joint representation of information state and distance $|C_t-S_t|$ between creature and source for the ITC given by Equation \ref{eq:joint-trans-ITC}.  Once this creature enters information state $\ISm{0}$ it remains in that state, with only its position relative to the injector varying.  The creature following the IT algorithm moves on each time step to the location either one to the right or one to the left of the injector's last know position, because these locations minimize the expected uncertainty on the following time step, and maximize the expected food gained amongst locations with equal expected uncertainty.  In addition, given two locations jointly optimal for information (first) and food (second) the ITC choses the closer of the two destinations.  There are three cases to consider.  
\begin{enumerate}
\item With probability $1-2x$ the source remains in place, or $S_{t+1}=S_t$.  In this case the creature remains in place as well, or $C_{t+1}=C_t$, and the distance moved is $\Delta_t=|C_{t+1}-C_t|=0$.  
\item With probability $x$ the source moves clockwise, or $S_{t+1}=S_t+1$.  
\begin{enumerate}
\item With probability $x/2$, $C_t=S_t-1$.  In this case, the creature moves one step CW, and $C_{t+1}=C_t+1$.  
\item With probability $x/2$, $C_t=S_t+1$.  In this case the injector moves (by chance) to colocalize with the creature.  The creature then moves either CW or CCW with equal likelihood, in order to remain in state $\ISm{0}$.  Therefore
\begin{enumerate}
\item With probability $x/4$, $C_{t+1}=C_t+1$.
\item With probability $x/4$, $C_{t+1}=C_t-1$.  
\end{enumerate}
\end{enumerate}
\item With probability $x$ the source moves counterclockwise, or $S_{t+1}=S_t-1$.
\begin{enumerate}
\item With probability $x/2$, $C_t=S_t+1$.  In this case, the creature moves one step CCW, and $C_{t+1}=C_t-1$.  
\item With probability $x/2$, $C_t=S_t-1$.  The creature then moves either CW or CCW with equal likelihood, so:
\begin{enumerate}
\item With probability $x/4$, $C_{t+1}=C_t+1$.
\item With probability $x/4$, $C_{t+1}=C_t-1$.  
\end{enumerate}
\end{enumerate}
\end{enumerate}
In sum, the creature remains in information state $\ISm{0}$ with probability 1, and moves a distance $\Delta_t=1$ with probability $2x$, or a distance $\Delta_t=0$ with probability $1-2x$.  Compare the middle column of Table 2.5.  
The calculations for the MLC and mMLC strategies follow in a similar fashion from the transition probabilities given in Equations \ref{eq:joint-trans-mMLC}-\ref{eq:joint-trans-MLC}.


\bibliographystyle{plain}
\bibliography{math,PJT,infotheory,neuroscience}

\end{document}